\documentclass[hidelinks,12pt]{article}   
\usepackage{amsthm,amssymb,amsmath}
\usepackage{graphicx}
\usepackage{booktabs}
\usepackage{mathtools}
\usepackage{setspace}
\usepackage{tabularx}
\usepackage{color}
\usepackage{epstopdf}
\usepackage[titletoc]{appendix}
\usepackage{cite}
\usepackage{natbib} 
\usepackage{float} 

\usepackage{footmisc} 
\usepackage{multirow}
\usepackage{afterpage}
\usepackage{bm}
\usepackage{rotating}
\usepackage{centernot} 

\usepackage{enumerate} 
\usepackage[colorlinks=true]{hyperref}   
\hypersetup{
    linkcolor=blue,     
    citecolor=blue,      
    urlcolor=blue    
}

\usepackage{tikz}
\usetikzlibrary{calc}
\usetikzlibrary{decorations.pathreplacing}
\usetikzlibrary{matrix}

\newtheorem{assumption}{Assumption}[section]

\newtheorem{lemma}{Lemma}[section]
\newtheorem{theorem}[lemma]{Theorem}
\newtheorem{corollary}[lemma]{Corollary}
\newtheorem{proposition}[lemma]{Proposition}

\theoremstyle{definition}
\newtheorem{remark}{Remark}[section]

\DeclareMathOperator*{\med}{med}
\DeclareMathOperator*{\sgn}{sgn}

\begin{document}

\title{Using SVM to Estimate and Predict \\ Binary Choice Models\thanks{Preliminary versions of this paper have been circulated since 2022. We thank Ye Lu and seminar and conference participants of the 2023 SoFie Conference in Seoul, Korea, the 2023 MEG Conference at Cleveland Fed, the AiE Conference and Festschrift in Honor of Joon Y. Park, the 2023 Melbourne Econometrics Workshop, Monash University, University of Essex, and 2025 Yale Econometrics Conference in Honor of Don Andrews for their helpful comments. We thank the 2024 Eminent Research Scholarship, University of Melbourne, for funding Joon Y. Park and Yoosoon Chang to visit the University of Melbourne.  
}\medskip}
\date{\today}

\singlespacing
\author{Yoosoon Chang\thanks{Department of Economics,  Indiana University, \href{mailto:yoosoon@iu.edu}{yoosoon@iu.edu}}
\and Joon Y. Park\thanks{Department of Economics,  Indiana University, \href{mailto:joon@iu.edu}{joon@iu.edu}}
\and Guo Yan\thanks{Department of Economics, University of Melbourne, \href{mailto:yan.g@unimelb.edu.au}{yan.g@unimelb.edu.au}} } 
\maketitle

\begin{abstract}\vspace{0.04in}
The support vector machine (SVM) has an asymptotic behavior that parallels that of the quasi-maximum likelihood estimator (QMLE) for binary outcomes generated by a binary choice model (BCM), although it is not a QMLE. We show that, under the linear conditional mean condition for covariates given the systematic component used in the QMLE slope consistency literature, the slope of the separating hyperplane given by the SVM consistently estimates the BCM slope parameter, as long as the class weight is used as required when binary outcomes are severely imbalanced. The SVM slope estimator is asymptotically equivalent to that of logistic regression in this sense. The finite-sample performance of the two estimators can be quite distinct depending on the distributions of covariates and errors, but neither dominates the other. The intercept parameter of the BCM can be consistently estimated once a consistent estimator of its slope parameter is obtained. 
\end{abstract}
\vfill

\singlespacing

\newpage
\renewcommand{\thefootnote}{\arabic{footnote}}
\setcounter{footnote}{0}
\onehalfspacing

\section{Introduction}

Machine learning has become increasingly popular in economics. See, e.g., \cite{mullainathan-spiess-17} and \cite{athey-imbens-19}, among many others. In machine learning, a variety of algorithms have been developed to analyze data within a flexible framework without an underlying statistical model. Although they were initially proposed as purely computational algorithms, it has recently been shown that some of them can be used to estimate a well-defined class of statistical and econometric models. For instance, it is now widely recognized that random forests and neural networks can be used to consistently estimate nonparametric regression functions under suitable regularity conditions. For the related asymptotic theories of random forests and neural networks, the reader is referred to, e.g., \cite{chi-vossler-fan-lv-22} and \cite{farrell-liang-misra-19}, respectively, and the references therein.

In this paper, we show that the support vector machine (SVM), another widely used algorithm for classification in machine learning, can be used under appropriate conditions to consistently estimate the binary choice model (BCM). The SVM was originally developed by \cite{cortes-vapnik-95} as an algorithm to classify binary outcomes observed with other related features. It is based on a simple idea of using a hyperplane to partition the feature space into two subsets associated with binary outcomes and has thus been regarded as an algorithm entirely independent of any statistical or econometric model. However, it is natural to consider the SVM as a method for analyzing the BCM, since the way the SVM classifies binary observations is essentially the same as that of a BCM with a zero median error term. Indeed, the SVM and the traditional approaches to BCMs use the same type of classifier, the one given by a separating hyperplane, to predict binary outcomes from covariates. For observations generated by the BCM, the SVM classifier will therefore not be asymptotically optimal unless it consistently estimates the underlying BCM.

The main finding of this paper is that the slope coefficient of the hyperplane provided by the SVM is a consistent estimator for the slope parameter in the BCM under the same conditions as those required for consistency of the quasi-maximum likelihood estimator (QMLE). The estimator provided by the SVM, which is referred to as the SVM estimator in the paper, is thus asymptotically equivalent to the QMLE in this sense.\footnote{The SVM estimator and the QMLE generally do not have the same limit distribution, so they are not fully equivalent even asymptotically. Also, our asymptotic setup does not allow for the number $m$ of features or covariates going to infinity.} This is rather surprising, since the SVM has been considered as an algorithm which does not rely on any statistical or econometric model. Given that the SVM is used widely to analyze binary outcomes in machine learning, this asymptotic equivalence of the SVM estimator to the QMLE adds an interesting and important interpretation of the SVM. We establish consistency of the SVM estimator for the slope parameter without assuming that a pseudo-true value exists. This is essential because, for the SVM estimator, a pseudo-true value may not exist if the two classes defined by a binary outcome are severely imbalanced, in which case the SVM fails and simply predicts a majority class.\footnote{This is well known. As shown in the paper, the problem disappears if we use \emph{class weight}, an option for the SVM adopted widely to deal with imbalanced data.}

The QMLE for the parameter in the BCM is generally inconsistent. It is, however, expected to yield a consistent estimator for the slope parameter in the BCM if the conditional expectation of covariates on its systematic component is given as a linear function. This condition is not as stringent as it may sound. It holds if covariates jointly have an elliptical distribution, and we may give appropriate weights to the observations on covariates so that they may be regarded as being drawn from such a distribution. See, e.g., \cite{ruud-86}, \cite{newey-ruud-94}. Consistent estimation of the slope parameter is of primary interest for the BCM, especially when its dimension is large with many covariates. This is because we may easily estimate the remaining one-dimensional intercept parameter once a consistent estimator for the multi-dimensional slope parameter is obtained. Under the required linearity in conditional expectation, consistency of the slope parameter in the BCM may be easily established if we assume the existence of a pseudo-true value as in \cite{ruud-83,ruud-86}.\footnote{The existence of a pseudo-true value is not guaranteed. For a rigorous proof of his result, see \cite{chang-park-yan-25}.}

The SVM estimator is comparable in many respects to the logit estimator, which is most commonly used QMLE especially for the large dimensional BCM, except that it breaks down in the extreme case where severe class imbalance exists. By simulation, we compare the SVM estimator with the logit estimator in terms of their finite sample performance. In finite samples, the performances of the two estimators could be quite distinct depending upon the distributions of covariates and errors. It seems clear that no one dominates the other. Though their performances are largely comparable, the SVM estimator performs significantly better than the logit estimator in some cases and vice versa in others.

Once a consistent estimator for the slope parameter is obtained, we may consistently estimate the intercept parameter in various ways. We may estimate the intercept parameter using, for example, the maximum score estimator of \cite{manski-85} or its smoothed version given by \cite{horowitz-92}, by replacing the slope parameter with its consistent estimate. In fact, we show in the paper that replacing the true slope parameter by its SVM estimator does not affect the limit distribution of the maximum score estimator for the intercept parameter.

The rest of the paper is organized as follows. In Section 2, we introduce the SVM and BCM, which is followed by a framework to analyze the classification of the BCM outcomes by the SVM. Section 3 provides the asymptotic properties of the SVM estimator, and discusses the conditions under which it is consistent for the slope parameter in the BCM. 
It also explains how to obtain a consistent estimator for the intercept parameter, once a consistent estimator for the slope parameter is available, in the BCM.
In Section 4, we contrast the SVM with the QMLE. We show that the SVM cannot be defined as a QMLE, even though the SVM and QMLE are largely comparable. Moreover, the effect of severe imbalance in binary outcomes is fatal to the SVM unless it is applied with class weights, while its effect is not so detrimental for the QMLE. Section 5 concludes the paper. Mathematical proofs are provided in the Appendix.

A word on notation. 
We let $\mathcal L(Z)$ signify the distribution of $Z$ for a random vector $Z$, and $\mathcal L(Z\mid W=w)$ signify the distribution of $Z$ conditional on another random vector $W$ taking the value $w$. Other standard notations include $1\{\cdot\}$ for the indicator function, $\|\cdot\|$ for the Euclidean norm, $\to_p$ for convergence in probability, and $o_p(1)$ as a generic notation for a term converging in probability to zero, among others, which will also be used. 
We also use the notation $f(\infty) := \lim_{t\to \infty} f(t)$ or $f(-\infty):=\lim_{t\to -\infty} f(t)$ when the limit is well-defined. 
These notations are used throughout the paper without any further reference.

\section{Preliminaries}\label{section-preliminaries}

Let $Y$ be a binary random variable taking values $1$ and $-1$, and $X$ be an $m$-dimensional random covariate which has the distribution $\mathcal L(X)$ over the support $\mathcal{X}$. We define $(Y_i,X_i)$, $i=1,\ldots,n$, to be an i.i.d.\ random sample from $(Y,X)$.

\subsection{Support Vector Machine}

The \emph{support vector machine} (SVM) is used widely in machine learning for classification problems.\footnote{We only consider the prototypical SVM without kernel trick throughout this paper; this restriction is mainly intended to relate} the SVM to the binary choice model.
The methodology does not assume any underlying model, and it can be used to analyze any binary-valued $(Y_i)$, $Y_i=1$ or $-1$, associated with an $m$-dimensional set of features $(X_i)$ for $i=1,\ldots,n$. More precisely, the SVM finds a hyperplane
\begin{equation}\label{svm-hyperplane}
\mathcal P = \left\{x\in\mathbb R^m\big|\alpha + x'\beta = 0\right\}
\end{equation}
defined by $\alpha\in\mathbb R$ and $\beta\in\mathbb R^m$, which most effectively separates the two subsets of $(X_i)$, one associated with $Y_i = 1$ and the other with $Y_i = -1$, depending upon which side of the hyperplane each $X_i$ falls on.

The SVM was first introduced in the case where $(X_i)$ coupled with $Y_i=1$ and $(X_i)$ coupled with $Y_i=-1$ can be separated in the space $\mathbb R^m$ by a hyperplane. In this case, the SVM method, which will be referred to as the \emph{hard-margin SVM}, chooses a hyperplane defined by $\alpha\in\mathbb R$ and $\beta\in\mathbb R^m$ given by the maximization problem
\begin{equation}\label{svm-hard-margin}
\max_{\alpha,\beta}\ M \quad \textrm{subject to} \quad
\frac{Y_i(\alpha + X_i'\beta)}{\|\beta\|} \geq M
\end{equation}
for all $i=1,\ldots,n$. Provided that $Y_i = \sgn(\alpha + X_i'\beta)$, $Y_i(\alpha + X_i'\beta)/\|\beta\| = |\alpha + X_i'\beta|/\|\beta\|$ is the distance of $X_i$ to the hyperplane. We call $M$ in \eqref{svm-hard-margin} the \emph{margin} between the hyperplane and $(X_i)$. Clearly, we may positively scale $\alpha$ and $\beta$ arbitrarily in \eqref{svm-hard-margin}, and thus, we may set $\|\beta\| = 1/M$ so that \eqref{svm-hard-margin} reduces to
\begin{align}\label{svm-hard-margin-aux}
\min_{\alpha,\beta}\ \frac{1}{2}\|\beta\|^2 \quad \textrm{subject to} \quad  Y_i(\alpha + X_i'\beta) \geq 1
\end{align}
for all $i=1,\ldots,n$. The hard-margin SVM is usually defined by \eqref{svm-hard-margin-aux}.

The two subsets of $(X_i)$, which are associated with $Y_i = 1$ and $Y_i=-1$, respectively, are rarely separable in practical applications. In this case, we use the \emph{soft-margin SVM}, which introduces slack variables $(\tau_i)$ and solves
\begin{equation}\label{svm-soft-margin}
\max_{\alpha,\beta,(\tau_i)}\ M \quad \textrm{subject to} \quad \frac{Y_i(\alpha + X_i'\beta)}{\|\beta\|} \geq M(1-\tau_i),\ \tau_i\ge 0
\end{equation}
for all $i=1,\ldots,n$ in the presence of an additional constraint $\sum_{i=1}^n\tau_i < \lambda_n$ with some constant $\lambda_n>0$ replacing \eqref{svm-hard-margin}. As before, we may rewrite \eqref{svm-soft-margin} as
\begin{equation}\label{svm-soft-margin-alt}
\min_{\alpha,\beta,(\tau_i)}\ \frac 12\|\beta\|^2 + \lambda_n\sum_{i=1}^n\tau_i
\quad\mbox{subject to}\quad
Y_i(\alpha + X_i'\beta)\ge 1-\tau_i,\ \tau_i\ge 0,
\end{equation}
for $i=1,\ldots,n$, by taking the additional constraint into consideration and redefining $\lambda_n$ appropriately. Once again, we may redefine $\lambda_n$ appropriately and reformulate the minimization in \eqref{svm-soft-margin-alt} as
\begin{equation}\label{svm-soft-margin-aux}
\min_{\alpha,\beta}\ \sum_{i=1}^n \big[1- Y_i(\alpha + X_i'\beta) \big]_+ +\lambda_n\|\beta\|^2,
\end{equation}
where $[\,\cdot\,]_+ = \max\,\{\cdot,0\}$. Note that the two constraints appearing in \eqref{svm-soft-margin-alt} may be combined and rewritten as
\[
\tau_i \ge \big[1-Y_i(\alpha + X_i'\beta)\big]_+
\]
and we set
\[
 \tau_i = \big[1-Y_i(\alpha + X_i'\beta)\big]_+
\]
to derive \eqref{svm-soft-margin-aux} from \eqref{svm-soft-margin-alt}. For the analysis of the SVM in this paper, we will focus on the soft-margin case, for which the optimization problem is given by \eqref{svm-soft-margin-aux}. 

\subsection{Binary Choice Model}  

A natural setting to generate a binary outcome is to use the binary choice model (BCM), which is given by 
\begin{equation}\label{bcm-linear}
 Y = \sgn\,(Y^\ast)
\quad\mbox{and}\quad
 Y^\ast = \alpha_0 + X'\beta_0 - U,
\end{equation}
where $\sgn$ is the sign function defined as $\sgn\,(z) = \pm 1$ for $z\ge 0$ and $z<0$, respectively, $X$ is an $m$-dimensional vector of covariates, $\theta_0 = (\alpha_0,\beta_0')'$ is the true value of the $(1+m)$-dimensional parameter vector $\theta = (\alpha,\beta')'$, and $U$ is the error term.  
Throughout the paper, we let $\theta\in\Theta$, where $\Theta$ is a convex subset of $\mathbb R^{1+m}$ that is arbitrarily large.\footnote{This is necessary because we identify $\theta_0$ only up to scalar multiplication in the derivation of our main result. Our definition of $\Theta$ is essentially equivalent to setting $\Theta = \mathbb R^{1+m}$. In this paper, however, we define $\Theta$ separately from $\mathbb R^{1+m}$ to make our subsequent assumptions more directly comparable to those existing in the literature.}

Since we focus on the consistency of the slope coefficient up to a positive scalar—referred to as slope consistency for simplicity—we impose conditions to ensure that $\theta_0$ is identified up to a positive scalar.

\begin{assumption}\label{aspn-med0-error} 
$\med\,(U|X) = 0$ almost surely in $\mathcal L(X)$. 
\end{assumption}

\begin{assumption}\label{aspn-pos-density-Xm}
(a) $X_m$ has a nonzero coefficient, and the distribution of $X_m$ conditional on $X_{-m}$ has everywhere positive Lebesgue density almost surely in $\mathcal L(X_{-m})$, 
where $X = (X_1,\ldots,X_m)'$ and $X_{-m} = (X_1,\ldots,X_{m-1})'$,
(b) $0 < \mathbb{P}\{Y=1|X \} <1$ almost surely in $\mathcal{L}(X)$, and (c) the support $\mathcal X$ of $\mathcal{L}(X)$ is not contained in any proper linear subspace of $\mathbb{R}^m$.
\end{assumption}

\noindent
Assumptions \ref{aspn-med0-error} and \ref{aspn-pos-density-Xm} are imposed for the BCM in \cite{horowitz-92}, modified from \cite{manski-75,manski-85}, which ensure that $\theta_0 = (\alpha_0,\beta_0')'$ is identified up to multiplication by a positive scalar. Both of them are assumed to hold throughout what follows.

\subsection{Classification of BCM Outcomes by SVM}

To predict a binary outcome $Y$, the SVM uses a classifier
\begin{equation}\label{classifier-svm-bcm}
 C(x) = \sgn\big(\alpha + x'\beta\big)
\end{equation}
with parameter $\theta = (\alpha,\beta')'$, which is formally defined as a function on $\mathcal X$ associating any realization of $X$ with a value of $Y$ taking values $\pm 1$. To evaluate a classifier $C$, we define its \emph{conditional risk} as 
\[
 R(x,C) = \mathbb P\big\{Y\ne C(X)\big|X=x\big\}
\]
for $x\in \mathcal X$, and its \emph{risk} as $R(C) = \int R(x,C)\,P(dx)$, where $P$ is the distribution of $X$. An \emph{optimal classifier} $C_0$ is defined as a classifier that has the smallest risk.

We can readily see that an optimal classifier is given by
\[
 C_0(x) = \sgn\big[\Pi(x)-1/2\big]
\]
for $x\in\mathcal X$, where $\Pi$ is the conditional choice probability (CCP) defined as $\Pi(x) = \mathbb P\big\{Y=1\big|X=x\big\}$ for $x\in\mathcal X$. This is because
\[
 R(x,C) = \Pi(x)1\big\{C(x) = -1\big\} + \big(1-\Pi(x)\big)1\big\{C(x) = 1\big\}
\]
is minimized if we set $C(x) = \pm 1$ for $x\in\mathcal X$ such that $\Pi(x)\gtrless 1/2$. 

For the BCM, the CCP is given by 
\[
 \Pi(x) = \mathbb P\big\{U<\alpha_0 + X'\beta_0\big|X=x\big\}
 = F_x\big(\alpha_0 + x'\beta_0\big),
\]
where $F_x$ is the distribution function of $U$ conditional on $X=x$ and $\theta_0 = (\alpha_0,\beta_0')'$ is the true value of the parameter $(\alpha,\beta')'$. However, under Assumption \ref{aspn-med0-error}, we have $F_x(u) \geq 1/2$ if and only if $u \geq 0$ and $F_x(u) < 1/2$ if and only if $u < 0$, which yields
\[
 \Pi(x) \begin{array}{c}\geq\\ <\end{array} 1/2 
\quad\mbox{if and only if}\quad 
 \alpha_0 + x'\beta_0 \begin{array}{c}\geq\\ <\end{array} 0 
\]
with $u = \alpha_0 + x'\beta_0$. The optimal classifier $C_0$ for the BCM is therefore given by
\begin{equation}\label{classifier-bcm}
 C_0(x) = \sgn\big(\alpha_0 + x'\beta_0\big),
\end{equation}
which belongs to the same class of classifiers as the SVM introduced in \eqref{classifier-svm-bcm}.

Suppose that observations are drawn from the BCM so that the optimal classifier is given by \eqref{classifier-bcm}, and that the SVM applied to those observations yields the classifier in \eqref{classifier-svm-bcm} with
\begin{equation}\label{theta-svm}
 \hat\theta = \big(\hat\alpha,\hat\beta'\big)',
\end{equation}
where $\hat\alpha$ and $\hat\beta$ define the separating hyperplane in the SVM. For the SVM classifier to be asymptotically optimal, $\hat\theta$ should converge in probability to $\theta_0$ up to a positive scale factor so that the separating hyperplane given by the SVM approaches the true hyperplane defined by $\theta_0$.\footnote{Here, asymptotic optimality of $\hat C$ means $\hat C(x) \to_p C_0(x)$ for $P$-almost every $x$.} Accordingly, we may use $\hat\theta$ as an estimator for $\theta_0$, recognizing that $\theta_0$ is identified only up to multiplication by a positive scalar.

For the rest of the paper, we assume that $(Y,X)$ is generated by the BCM unless stated otherwise explicitly, and consider $\hat\theta$ introduced in \eqref{theta-svm} as an estimator, called the \emph{SVM estimator}, of the parameter $\theta_0$ in the underlying BCM. As will be shown, $\hat\theta$ is consistent for $\theta_0$ under suitable regularity conditions. The required conditions for consistency of the SVM estimator $\hat\theta$, including both $\hat\alpha$ and $\hat\beta$ for the intercept and slope parameters $\alpha_0$ and $\beta_0$, respectively, are quite stringent. Consistency of the slope estimator $\hat\beta$ for the slope parameter $\beta_0$ is, however, much milder and only requires conditions that we typically impose for consistency of the quasi-maximum likelihood estimator (QMLE).


\section{Asymptotics for SVM Estimator}\label{section-svm-asymptotics}

In this section, we develop the asymptotic properties of the SVM estimator $\hat\theta = (\hat\alpha,\hat\beta')'$ for the underlying BCM model.

\subsection{Asymptotic Framework}

Before establishing the consistency of the SVM estimator $\hat\theta = (\hat\alpha,\hat\beta')'$ for the BCM parameter $\theta_0 = (\alpha_0,\beta_0')'$ under suitable conditions, we start by defining the probability limit $\theta_\ast = (\alpha_\ast,\beta_\ast')'$ of the SVM estimator without assuming that $(Y,X)$ is generated by the BCM. 

The SVM yields the solution $\hat\theta = (\hat\alpha,\hat\beta')'$ defined as a minimizer of
\begin{equation}\label{eq-Qn}
Q_n(\theta) = \frac{1}{n} \sum_{i=1}^n \big[1- Y_i(\alpha + X_i'\beta)\big]_+ +\frac{\lambda_n}{n} \|\beta\|^2,
\end{equation}
which is assumed to exist and be unique. 
For our asymptotics in the paper, we follow \cite{koo-lee-kim-park-08}, \cite{zhang-wu-wang-li-16} and \cite{wang-yang-chen-liu-19}, among others, and assume that $\lambda_n = o(n)$.\footnote{By setting $\lambda_n$ to be large, we make it less costly to increase the margin and allow for the observations in the margin and even on the other side. In practice, it is common to set $\lambda_n$ to a fixed value. For instance, when using existing packages in Python (\texttt{sklearn.svm.LinearSVC}) or MATLAB (\texttt{fitcsvm}), the default choice is $\lambda_n = 1/2$, under which conditions $\lambda_n = o(n)$ and $\lambda_n = o(\sqrt{n})$ used in this paper are satisfied. 
\label{footnote-lambda-n}
} Under this and other suitable regularity conditions, we may then expect that $Q_n(\theta)\to_p Q(\theta)$, where
\begin{equation}\label{eq-Q}
Q(\theta) =\mathbb{E}\big[1-Y(\alpha+X'\beta)\big]_+,
\end{equation}
uniformly in $\theta$ over some parameter set $\Theta$, from which it follows immediately that $\hat\theta \to_p \theta_\ast$ if $Q(\theta)$ has a unique minimum $\theta_\ast$.

\begin{assumption}\label{aspn-regularity} 
(a) $Q$ has a unique minimum $\theta_\ast$ which is an interior point of $\Theta$, 
(b) $\mathbb E \|X\| <\infty$. 
\end{assumption}

\begin{lemma}\label{lemma-plim-theta}
Let Assumption \ref{aspn-regularity} hold.  
Then $\hat\theta \to_p \theta_\ast$. 
\end{lemma}
\noindent 
Lemma \ref{lemma-plim-theta} follows immediately from Theorem 2.7 in \cite{newey-mcfadden-94}, since all conditions required in their theorem are trivially satisfied under our Assumption \ref{aspn-regularity}. In particular, note that $Q_n(\theta)$ and $Q(\theta)$ are convex in $\theta\in \Theta$.

Lemma \ref{lemma-plim-theta} assumes that $\theta_\ast$ is the unique minimizer of $Q(\theta)$ over $\Theta$. 
This assumption will not be imposed but verified under the conditions for the slope consistency introduced below. In particular, provided the slope consistency, i.e., $\beta_\ast = c_\ast \beta_0$ for some $c_\ast > 0$, Lemma \ref{lemma-gradient-hessian-svm} below shows that the Hessian matrix of $Q$ at $\theta_\ast$ is positive definite. It then follows that $\theta_\ast$ is the unique minimizer of $Q$.

\paragraph{Gradient and Hessian for SVM when well-defined } 
For the later analysis of SVM slope consistency, we will make use of the gradient of $Q$. Define 
\begin{align}\label{gradient-svm}
\dot{Q}(\theta) = - \mathbb{E}\,1\big\{1-Y(\alpha+X'\beta) >0\big\}\,
Y\!\begin{pmatrix} 1 \\ X \end{pmatrix},
\end{align}
and 
\begin{align}\label{hessian-svm}
\ddot{Q}(\theta) =  \mathbb{E}\,\delta\big(1-Y(\alpha+X'\beta)\big)  \begin{pmatrix}
1 & X' \\ X & XX'
\end{pmatrix},
\end{align}
where $\delta(\cdot)$ denotes the Dirac delta function with $\delta(z) = (d/dz) 1\{z>0\}$ in the distributional sense. For any $\theta = (\alpha,\beta')'\in \mathbb R^{1+m}$ such that $\beta_m\neq 0$, $\alpha+X'\beta$ has a Lebesgue density $f_{\alpha + X'\beta}$ by Assumption \ref{aspn-pos-density-Xm} (a), and we may thus define $\ddot Q(\theta)$ more explicitly as
\begin{align*}
\ddot{Q}(\theta) & = f_{\alpha+X'\beta}(1) \mathbb{E}\, \left[ 1\{Y = 1 \} \begin{pmatrix} 1 & X' \\ X & XX' \end{pmatrix}
\bigg| \alpha + X'\beta = 1 \right] \\
& \quad + f_{\alpha+X'\beta}(-1) \mathbb{E}\, \left[ 1\{Y = -1 \} \begin{pmatrix} 1 & X' \\ X & XX' \end{pmatrix}
\bigg| \alpha + X'\beta = -1 \right]. 
\end{align*}

Under some regularity conditions, which will be provided below, $\dot{Q}(\theta)$ and $\ddot{Q}(\theta)$ actually become the gradient vector and Hessian matrix of $Q(\theta)$, respectively, for $\theta\in\mathbb R^{1+m}$ such that $\beta_m \neq 0$.  

\begin{assumption}\label{aspn-svm-hessian}
We assume that (a) $\mathbb E\|X\|^2<\infty$, (b) the Lebesgue density $p(x_m|x_{-m})$ of $\mathcal{L}(X_m|X_{-m}=x_{-m})$ is uniformly bounded, i.e., $p(x_m|x_{-m}) < K$ for some constant $K$ independent of $(x_m,x_{-m})$, and $p(x_m|x_{-m})$ is continuous in $x_m$, and (c) $\mathbb P\big\{Y=1\big|X=x\big\}$ is continuous in $x_m$.
\end{assumption}

\noindent
If $(Y,X)$ is generated by the BCM, Assumption \ref{aspn-svm-hessian} (c) requires the continuity of the CCP $\Pi(x) = F_{x}(\alpha_0+x'\beta_0)$ in $x_m$, where $F_{x}$ is the distribution function of $U$ conditional on $X=x$ and $x=(x_{-m}',x_m)'$.

\begin{lemma}\label{lemma-gradient-hessian-svm} 
Let $\theta = (\alpha,\beta')' \in \mathbb R^{1+m}$ be arbitrary with $\beta_m \neq 0$. 
Under Assumption \ref{aspn-pos-density-Xm} (a),\footnote{Here, we only need the existence of a Lebesgue density, without requiring it to be positive almost everywhere. 
}
the gradient vector of $Q$ at $\theta$ is given by $\dot{Q}(\theta)$. Under Assumptions \ref{aspn-pos-density-Xm} and \ref{aspn-svm-hessian}, the Hessian matrix of $Q$ at $\theta$ is given by $\ddot{Q}(\theta)$, and $\ddot{Q}(\theta)$ is continuous and positive definite.
\end{lemma}

\noindent
Lemma \ref{lemma-gradient-hessian-svm} was established earlier by \cite{koo-lee-kim-park-08} and has been used in later work on SVMs, such as \cite{zhang-wu-wang-li-16}. Their results, however, require that $\mathcal{L}(X|Y=1)$ and $\mathcal{L}(X|Y=-1)$ have Lebesgue densities, which is rather restrictive since it does not allow for the presence of discrete variables in $X$. In contrast, our assumptions do not impose such restrictions on $(Y,X)$, requiring only that the scalar index admit a Lebesgue density.

\subsection{Consistency of SVM Estimator}\label{section-consistency}

Subsequently, we assume that $(Y,X)$ is generated by the BCM with the true parameter value $\theta_0$. The SVM slope $\hat\beta$ is consistent for $\beta_0$ if and only if the probability limit $\beta_\ast$ of $\hat\beta$ is a positive multiple of $\beta_0$, and $\hat\theta$ is consistent for $\theta_0$ if and only if $\theta_\ast$ is a positive multiple of $\theta_0$. In this section, we consider the conditions for slope consistency and consistency of the SVM estimator, with a main focus on the slope consistency.  

Let 
\[
V = \alpha_0 + X'\beta_0
\]
and 
\[
\Pi(v) = \mathbb P\{Y=1 |V=v\}  = \mathbb P\{U\leq V|V=v\}. 
\]
We introduce two important assumptions that are crucial for the slope consistency of the SVM, which are also imposed for the slope consistency of the QMLE as in \cite{chang-park-yan-25}. 

\begin{assumption}\label{id-dep} $\mathcal L(U|X) = \mathcal L(U|V)$.
\end{assumption}

\noindent
Assumption \ref{id-dep} requires that the error distribution depends on $X$ only through the index $V$. Such an \emph{index dependence} condition for the error distribution is often used in the literature. It is used in, e.g., \cite{klein-spady-93} for the BCM, and \cite{ichimura-93} for a class of single-index models, among many others. This condition is satisfied trivially if $U$ is independent of $X$, as is often assumed for the BCM.

\begin{assumption}\label{ld-mean} 
$\mathbb E(X|V) = aV + b$ for some $a,b \in \mathbb R^m$.
\end{assumption}

\noindent
Assumption \ref{ld-mean} requires that the conditional mean of $X$ given $V$ is a linear function of $V$, which will be referred to simply as the \emph{linearity in expectation} condition. This condition is certainly restrictive, but it holds, for instance, if $X$ has an elliptical distribution. 
We may also weight observations of $X$ appropriately so that the distribution of the weighted observations can be regarded as satisfying this condition; 
see, e.g., \cite{ruud-86} and \cite{newey-ruud-94}.\footnote{
Specifically, in place of the original observations $\{x_i\}_{i=1}^n$ of $X$, they use the observations weighted by $w_i = \sigma(x_i)/\hat\tau(x_i)$ for $i=1,\ldots,n$, where $\sigma$ is the standard multivariate normal density which satisfies Assumption~\ref{ld-mean}, and $\hat\tau$ is a local kernel density estimator of $X$. The weighted observations may be approximately regarded drawn from a distribution with density $\sigma$.  
} 

To analyze the SVM slope consistency, we consider the SVM over a restricted set of parameters $\theta = (\alpha,\beta')'$ specified as 
\begin{equation}\label{res}
 \begin{pmatrix} \alpha \\ \beta \end{pmatrix} 
 = c \begin{pmatrix} \alpha_0 \\ \beta_0 \end{pmatrix} 
 + \begin{pmatrix} r \\ 0 \end{pmatrix} 
\end{equation}
with $c,r\in \mathbb R$. The resulting estimator will be referred to as the \emph{restricted} SVM, whose population objective function is 
\begin{align*}
Q_\bullet (c,r) = \mathbb E \big[1- Y(cV + r)\big]_+. 
\end{align*}
Note that by Lemma~\ref{lemma-gradient-hessian-svm}, for any $c\neq 0, r\in \mathbb R$, the gradient of $Q_\bullet(c,r)$ is 
\begin{align*}
\dot Q_\bullet(c,r) & = - \mathbb{E}\Big(  1\{U\leq V,\,c  V+r < 1\}  - 1\{U>V,\,c V+r > -1\} \Big)\begin{pmatrix} 1 \\ V \end{pmatrix} \\
& = - \mathbb{E}\Big( \Pi(V) 1\{c V+r < 1\}  - (1-\Pi(V) ) 1\{cV+r > -1\} \Big)\begin{pmatrix} 1 \\ V \end{pmatrix} .
\end{align*} 
Let 
\begin{align*}
(c_\ast,r_\ast) = \arg\min_{c,r} Q_\bullet (c,r). 
\end{align*} 
Provided that $c_\ast > 0$, $(c_\ast,r_\ast)$ solves the FOCs $\dot Q_\bullet(c,r)=0$, and thus, 
\begin{align}\label{dotQ-dotQbullet}
& \dot Q\left( c_\ast \begin{pmatrix} \alpha_0 \\ \beta_0 \end{pmatrix} 
 + \begin{pmatrix} r_\ast \\ 0 \end{pmatrix} \right)  \notag  \\
& = - \mathbb{E}\Big(  1\{U\leq V,\,c_\ast V+r_\ast  < 1\}  - 1\{U>V,\,c_\ast  V+r_\ast  > -1\} \Big)\begin{pmatrix} 1 \\ X \end{pmatrix}  \notag \\
&  = - \mathbb{E}\Big(  \Pi(V) 1\{c_\ast V+r_\ast < 1\}  - (1-\Pi(V) ) 1\{c_\ast V+r_\ast  > -1\} \Big)\begin{pmatrix} 1 \\ \mathbb E(X|V) \end{pmatrix}  \notag  \\
& = - \mathbb{E}\Big(  \Pi(V) 1\{c_\ast V+r_\ast < 1\}  - (1-\Pi(V) ) 1\{c_\ast V+r_\ast > -1\} \Big)\begin{pmatrix} 1 \\ aV +b \end{pmatrix}  \notag \\
& = 0 
\end{align}
where the second equality follows from the law of iterated expectation and Assumption~\ref{id-dep}, the third equality is by Assumption \ref{ld-mean}, and the last line is by $\dot Q_\bullet(c,r) =0$. 
Since $Q(\theta)$ is convex and has a positive definite Hessian matrix at any $\theta$ with $\beta_m \neq 0$ by Lemma \ref{lemma-gradient-hessian-svm}, it holds that 
\begin{align*}
 \begin{pmatrix} \alpha_\ast \\ \beta_\ast \end{pmatrix} 
 = c_\ast \begin{pmatrix} \alpha_0 \\ \beta_0 \end{pmatrix} 
 + \begin{pmatrix} r_\ast \\ 0 \end{pmatrix} ,
\end{align*}
and thus, the SVM slope $\hat\beta$ is consistent for $\beta_0$.

It is clear that the key to establishing the slope consistency of the SVM is to ensure $c_\ast>0$. In contrast to the slope consistency of the QMLE analyzed in \cite{chang-park-yan-25}, we need an additional condition to ensure $c_\ast >0$. 

In what follows, we first introduce some notations, after which Lemma \ref{lemma-cast-positive} provides the necessary and sufficient conditions for $c_\ast>0$. 
Let $\mathbb P\{Y=1\} \ge 1/2$. This causes no loss of generality, since we may redefine $Y$ by $-Y$ if necessary. Define 
\[
 p(v) = \mathbb P\big\{U \leq V<v\big\}
\quad\mbox{and}\quad
 q(v) = \mathbb P\big\{U>V>v\big\},\vspace{-0.084in}
\]
\[
 \tau(v) = \mathbb E\,V1\big\{U \leq V<v\big\}
\quad\mbox{and}\quad
 \sigma(v) = \mathbb E\,V1\big\{U>V>v\big\}
\]
for $v\in\mathbb R$.  
Under Assumption \ref{aspn-pos-density-Xm} (a)-(b), $p$ is strictly increasing with $p(-\infty) = 0$ and $p(\infty) = \mathbb P\{Y = 1\} \ge 1/2$, whereas $q$ is strictly decreasing with $q(-\infty) = \mathbb P\{Y=-1\} \le 1/2$. Thus, $(p^{-1}\circ q)(-\infty)$  is well-defined in $\mathbb R \cup\{\infty\}$.

\begin{assumption}\label{aspn-suff-cond-consistency-domclass}
We assume that $\tau(\overline v) > \sigma(-\infty)$, where $\overline v = (p^{-1}\circ q)(-\infty)$.
\end{assumption}

The following lemma shows that Assumption \ref{aspn-suff-cond-consistency-domclass} is indeed the necessary and sufficient condition for $c_\ast>0$, under some regularity conditions. 
We discuss the interpretation of Assumption \ref{aspn-suff-cond-consistency-domclass} in Section \ref{section-balance-illustration}.

\begin{lemma}\label{lemma-cast-positive}  
Let Assumptions \ref{aspn-med0-error}, \ref{aspn-pos-density-Xm}  
hold. 
Then minimizing $Q_\bullet (c,r)$ over $c\in \mathbb R, r \in \mathbb R$ has a unique solution at some $c_\ast>0, r_\ast \in \mathbb R$ 
if and only if Assumption \ref{aspn-suff-cond-consistency-domclass} is satisfied. 
\end{lemma}

\begin{theorem}\label{theorem-consistency-svm}
Let $(Y,X)$ be generated by the BCM satisfying Assumptions \ref{aspn-med0-error}, \ref{aspn-pos-density-Xm}, \ref{aspn-svm-hessian}, \ref{id-dep}, \ref{ld-mean}, and \ref{aspn-suff-cond-consistency-domclass}. Then the SVM slope estimator $\hat\beta$ is consistent for $\beta_0$ (up to multiplication by a positive scalar).
\end{theorem}

\noindent
Theorem~\ref{theorem-consistency-svm} establishes the consistency of the SVM slope estimator $\hat\beta$ for $\beta_0$ in the BCM. Corollary~\ref{corollary-consistency-svm} further shows that, under an additional symmetry condition, the intercept estimator $\hat\alpha$ is also consistent, together with $\hat\beta$. 

\begin{corollary}\label{corollary-consistency-svm}
Suppose that the assumptions in Theorem \ref{theorem-consistency-svm} hold. If $\mathcal L(U,V) = \mathcal L(-U,-V)$, then the SVM intercept estimator $\hat\alpha$ and the SVM slope estimator $\hat\beta$ are both consistent for $\alpha_0$ and $\beta_0$ (up to multiplication by a positive scalar), respectively.
\end{corollary}

\subsection{Estimation of Intercept Given Consistent SVM Slope Estimator} 

Provided that the SVM slope estimator $\hat \beta$ is consistent up to a scalar $c_\ast >0$, we consider the estimation of the $\alpha_0$ up to the same scalar in this subsection. 

Define a second-step estimator $\hat \alpha_{ms}$ by
\begin{equation}\label{alpha-ms}
\hat \alpha_{ms} = \arg \max_{\alpha\in \mathcal A} \hat Q_{ms}(\alpha,\hat \beta) , 
\end{equation}
where 
\begin{equation}\label{Q-hat-ms-beta}
\hat Q_{ms}(\alpha,\beta) = \frac{1}{n} \sum_{i=1}^n Y_i 1\{ \alpha+ \beta'X_i \geq 0 \}  
\end{equation}
and $\mathcal A$ is a compact and convex set on $\mathbb R$.  
Note that $\hat\alpha_{ms}$ is a maximum score estimator with a univariate generated regressor $\hat \beta'X_i$ using the SVM slope estimator $\hat\beta$.  
In particular, \cite{manski-75,manski-85} show that the maximum of
\[
Q_{ms}(\alpha,\beta) = \mathbb E Y 1\{ \alpha+ \beta'X \geq 0 \}  
\]
is given by $c\alpha_0, c\beta_0$ for any constant $c>0$. Therefore, if we fix $\beta$ at $\beta_\ast = c_\ast\beta_0$, then $Q_{ms}(\alpha,\beta_\ast)$ has a unique maximum $\alpha_\ast =c_\ast \alpha_0$ for $\alpha$.

Since $\hat \beta$ is consistent for $\beta_0$ up to a scalar $c_\ast >0$, it is expected that $\hat \alpha_{ms}$ has probability limit $c_\ast \alpha_0$. Consequently, $(\hat \alpha_{ms}, \hat \beta')'$ is consistent for $c_\ast \theta_0$.  The following proposition establishes the consistency of $\hat \alpha_{ms}$.

\begin{proposition}\label{proposition-alpha-ms-consistency}
Let the conditions in Theorem \ref{theorem-consistency-svm} hold so that $\hat \beta \to_p c_\ast \beta_0$ for some $c_\ast >0$. 
Assume that $\mathcal A$ is a compact and convex set on $\mathbb R$ and $c_\ast\alpha_0$ is an interior point of  $\mathcal A$. 
Then 
\[
\hat \alpha_{ms} \to_p c_\ast \alpha_0.
\]
\end{proposition}

To establish the convergence rate and limit distribution of $\hat\alpha_{ms}$, we introduce some assumptions. 
Assume $\beta_{m,0}>0$ without loss of generality, since we can use $-X_m$ in place of $X_m$. 

\begin{assumption}\label{aspn-alpha-ms-rate} 
(a) $\mathcal A$ is a compact and convex set on $\mathbb R$, and $c_\ast \alpha_0$ is an interior point of $\mathcal A$; 
(b) $\tilde Q_{ms}(\alpha,\beta_{-m}) :=  \mathbb E Y 1\{ \alpha+ \beta_{-m}'X_{-m} + X_{m}  \geq 0 \}$  is twice continuously differentiable with negative definite Hessian matrix at $\big(\frac{\alpha_0}{\beta_{m,0}},\frac{\beta_{-m,0}'}{\beta_{m,0}} \big)'$; 
(c) The support of $X_{-m}$ is bounded;  
(d) There exists a constant $K>0$ such that, the Lebesgue density $p(x_m|x_{-m})$ of $\mathcal{L}(X_m|X_{-m}=x_{-m})$ and its derivatives up to second order are uniformly bounded, i.e., $p(x_m|x_{-m}), \nabla_{x_m} p(x_m|x_{-m}),  \nabla^2_{x_m} p(x_m|x_{-m})< K$ for a constant $K$ independent of $(x_m,x_{-m})$. In addition, the first and second derivatives of the distribution function $F_{x}(\cdot)$ of $\mathcal{L}(U|X=x)$ are uniformly bounded by $K$; 
(e) $\mathbb E \|X\|^3 <\infty$. 
\end{assumption}

Conditions in Assumption \ref{aspn-alpha-ms-rate} are also assumed in \cite{chen-lee-sung-14}, which studies maximum score (MS henceforth) estimation when some regressors are conditional expectations and need to be estimated nonparametrically in the first stage. \cite{chen-lee-sung-14} show that the asymptotic distribution of the second stage MS estimator is the same as if the conditional expectations were used, provided that the conditional expectations are estimated with a convergence rate faster than $n^{1/3}$ in the first stage, along with some other regularity conditions. 

If $\beta_\ast$ were known, $\alpha$ could be estimated by 
\[
\hat \alpha_\ast = \arg \min_{\alpha\in \mathcal A} \hat Q_{ms}(\alpha,\beta_\ast), 
\]
whose cube root convergence rate and limit distribution have been established in \cite{kim-pollard-90}. Now we give a result which states that the second stage MS estimator $\hat \alpha_{ms}$ using the SVM's slope has the same limit distribution as the infeasible estimator $\hat \alpha_\ast.$

\begin{theorem}\label{theorem-alpha-ms-equiv-limit-distribution}
Let the conditions in Theorem \ref{theorem-consistency-svm} hold so that $\hat \beta \to_p c_\ast \beta_0$ for some $c_\ast >0$. Let Assumption \ref{aspn-alpha-ms-rate} hold. Provided that $\|\hat \beta-\beta_\ast\| = O_p(n^{-1/2})$, then $n^{1/3}(\hat\alpha_{ms} - \alpha_\ast)$ is asymptotically equivalent in distribution to $n^{1/3}(\hat\alpha_{\ast} - \alpha_\ast)$. 
\end{theorem}

Theorem \ref{theorem-alpha-ms-equiv-limit-distribution} uses $\|\hat \beta-\beta_\ast\| = O_p(n^{-1/2})$, which will be derived in the next section under suitable regularity conditions with other asymptotic properties of the SVM estimator.

\subsection{Asymptotic Distributions}\label{section-limit-distn}

In this section, we present the convergence rate and the asymptotic distribution of the SVM estimator. We first provide a technical lemma. 

\begin{lemma}\label{lemma-svm-cgce-rate-aux}
Let Assumptions \ref{aspn-pos-density-Xm},  \ref{aspn-regularity}, \ref{aspn-svm-hessian} hold, and $\lambda_n =o(\sqrt{n})$.  
When $\beta_{m,\ast}\neq 0$, it holds that
\begin{align*}
\mathbb{E} \sup_{ \|\theta-\theta_\ast\| <\delta }  \Big| (Q_n - Q )(\theta) - (Q_n - Q )(\theta_\ast) \Big|   = O\big( n^{-1/2} \delta \big)
\end{align*}
as $\delta \to 0$ and $n\to \infty$. 
\end{lemma}

Lemma \ref{lemma-svm-cgce-rate-aux} shows that the expected modulus of continuity of the process $Q_n-Q$ around $\theta_\ast$, over a neighborhood of radius $\delta$, is of order $n^{-1/2} \delta$. This result relies on a Lipschitz condition of the mapping $\theta \mapsto \ell_{\theta}(y,x) := (1-y(\alpha+x'\beta))_{+}$, in the sense that $\big| \ell_{\theta_1}(y,x) - \ell_{\theta_2}(y,x) \big| \leq \sqrt{1+\|x\|^2} \|\theta_1 -\theta_2\|$ for any $(y,x)\in \{-1,1\}\times \mathbb R^{1+m}$ and any $\theta_1,\theta_2 \in \mathbb R^{1+m}$.\footnote{
For a class of functions Lipschitz in the index parameter in Euclidean space, its bracketing number can be upper bounded by the diameter of the index parameter set to the power of the dimension.  See, e.g., Example 19.7 in \cite{vandervaart-00}, or Theorem 2.7.11 in \cite{vandervaart-wellner-96} for a more general statement.
}
Assumptions \ref{aspn-pos-density-Xm} and \ref{aspn-svm-hessian}  guarantee that $Q(\theta)$ is twice continuously differentiable at $\theta_\ast$ with a positive definite Hessian matrix, as shown in Lemma \ref{lemma-gradient-hessian-svm}. 
These assumptions are also used to ensure that $\hat\theta \to_p \theta_\ast$ as in Lemma~\ref{lemma-plim-theta}. The only difference is that $\lambda_n = o(\sqrt{n})$ is assumed here, which is stronger than $\lambda_n = o(n)$ in Lemma~\ref{lemma-plim-theta}. Note that $\lambda_n = o(\sqrt{n})$ holds when using default choices in existing packages; see Footnote~\ref{footnote-lambda-n} for details.

Now we are ready to establish the $\sqrt{n}$ convergence rate of the SVM estimator $\hat \theta$ to its probability limit $\theta_\ast$. 

\begin{lemma}\label{lemma-svm-cgce-rate}
Let Assumptions \ref{aspn-pos-density-Xm}, \ref{aspn-regularity}, \ref{aspn-svm-hessian} hold, and $\lambda_n =o(\sqrt{n})$.  
When $\beta_{m,\ast}\neq 0$, it holds that
\[
\sqrt{n}(\hat\theta - \theta_\ast) = O_p(1)
\]
as $n\to \infty$. 
\end{lemma}

\begin{theorem}\label{theorem-svm-limit-distn}
Let Assumptions \ref{aspn-pos-density-Xm}, \ref{aspn-regularity}, \ref{aspn-svm-hessian} hold, and $\lambda_n =o(\sqrt{n})$.  
When $\beta_{m,\ast}\neq 0$, it holds that
\[
\sqrt{n}(\hat\theta-\theta_\ast) 
\to_d  N(0,\ddot{Q}(\theta_\ast)^{-1}J \ddot{Q}(\theta_\ast)^{-1}),
\]
where 
\[
J  =  \mathbb{E} 1\{1-  Y (\alpha_\ast + X'\beta_\ast) > 0\} \begin{pmatrix}
1 & X' \\ X & XX'
\end{pmatrix}
\]

\end{theorem}

The limit distribution is the same as that in \cite{koo-lee-kim-park-08}. However, one key assumption in \cite{koo-lee-kim-park-08} is that $\mathcal{L}(X|Y=\pm 1) $ have Lebesgue density functions which are continuous, whereas our results allow for discrete components of $X$.

\section{SVM and QMLE}

In the same BCM context, \cite{chang-park-yan-25} show that the QMLE, which specifies that $U$ is independent of $X$ and has distribution function $G$ such that $\log G, \log (1-G)$ are strictly concave, is slope consistent under Assumptions \ref{id-dep}, \ref{ld-mean}. 
However, the slope consistency of the QMLE does not require Assumption \ref{aspn-suff-cond-consistency-domclass}. 

This section contains four parts. First, we show that the SVM estimator is not a QMLE. Second, we interpret Assumption \ref{aspn-suff-cond-consistency-domclass}, which ensures the slope consistency of the SVM estimator, as a condition ruling out severe class imbalance. 
We provide an illustration to show that Assumption \ref{aspn-suff-cond-consistency-domclass} may fail under some specifications when, roughly speaking, two classes are imbalanced in the sense that $Y$ takes on one value more frequently than another. 
Third, we show that the failure of the SVM under severe class imbalance can be effectively addressed by using class-weighted SVM. Indeed, the class-weighted SVM is shown to be slope consistent under Assumptions \ref{id-dep} and \ref{ld-mean}, without requiring the non-severe imbalance condition in Assumption \ref{aspn-suff-cond-consistency-domclass}.
Lastly, we provide simulation studies to show that, when both SVM and QMLE are slope consistent, their finite sample performance may differ across specifications, with one potentially outperforming the other.

\subsection{Noninterpretability of SVM as QMLE}
In this section, we show that SVM's optimization problem is not equivalent to that of the quasi-maximum likelihood estimation (QMLE) for any specified distribution function. In this sense, the SVM cannot be viewed as a QMLE, even though both estimators are slope consistent under appropriate conditions.

The QMLE for BCM \eqref{bcm-linear}, which assumes that $U$ is independent of $X$ and has distribution function $G$, maximizes  
\[
Q_{qml}(\theta) = \mathbb{E}\big[ 1\{Y=1\} \log G(\alpha+X'\beta) + 1\{Y=-1\} \log(1-G(\alpha+X'\beta)) \big]
\]
in the limit, 
whereas the SVM minimizes
\begin{align*}
Q(\theta) & =\mathbb{E}\big[1-Y(\alpha+X'\beta)\big]_+ \\
& = \mathbb E \big[ 1\{Y=1\} (1-(\alpha+X'\beta) )_{+} + 1\{Y=-1\} (1+(\alpha+X'\beta) )_{+}  \big] .
\end{align*} 
For the SVM to be interpreted as the QMLE with a certain $G$, it is required that there exists some constant $\gamma >0$ such that
\begin{align*}
- \gamma (1-z)_{+} & = \log G(z)  \\
- \gamma (1+z)_{+} & = \log (1-G(z))
\end{align*} 
for all $z$. 
The first equality yields that $G(z) = e^{ - \gamma (1-z)_{+} }$ and the second yields $1-G(z) = e^{ - \gamma (1+z)_{+} }$. Unfortunately, there does not exist a constant $\gamma $ such that $e^{ - \gamma (1-z)_{+} } +  e^{ - \gamma (1+z)_{+} } = 1$ for all $z$, which gives a contradiction. 

\subsection{Differing Effects of Imbalance}\label{section-balance-illustration}

The condition in Assumption \ref{aspn-suff-cond-consistency-domclass} plays a very important role in the consistency of the SVM estimator. In this section, we discuss this condition and show that it can be understood as a \emph{non-severe imbalance condition}. 

We introduce a notion. 
For two probability measures $P$ and $Q$ on $\mathbb R$, $Q$ is said to dominate $P$ for $\Pi$, written as $P\prec_\Pi Q$, if $\int \Pi(v)P(dv) < \int \Pi(v)Q(dv)$.\footnote{If $Q$ stochastically dominates $P$ in the first order and $\Pi(v) = \mathbb P\{U\leq V|V=v\}$ is strictly increasing in $v$, then $Q$ always dominates $P$ for $\Pi$ in our sense.} We also say that $P$ decreases for $\Pi$ if $\int \Pi(v)P(dv)$ decreases as $P$ changes.  

\begin{lemma}\label{lemma-v-bar}
Let Assumption \ref{aspn-med0-error}, Assumption \ref{aspn-pos-density-Xm} (a)-(b) hold. 
We have (a) $\overline v >0$, (b) $\overline v$ increases up to $\infty$ as the classes get more balanced in the sense that $\mathcal L(V)$ decreases for $\Pi$ down to $1/2$, and (c) $\tau(\infty) > \sigma(-\infty)$.
\end{lemma}

\noindent
In this lemma, we show that $\overline v$ decreases on $\mathbb R_+$ as the classes get more imbalanced, from which it follows immediately that Assumption \ref{aspn-suff-cond-consistency-domclass} becomes harder to be satisfied since $\tau$ is monotone increasing and $\tau(\overline v)$ decreases monotonically on $\mathbb R_+$. 
In fact, as we will demonstrate later, Assumption \ref{aspn-suff-cond-consistency-domclass} is violated if the class imbalance becomes severe. If the classes are balanced, we have $\overline v=\infty$, and Assumption \ref{aspn-suff-cond-consistency-domclass} always holds.

\begin{remark}[Differing effects of class imbalance on the slope consistency of the SVM and the QMLE]
As established in Theorem \ref{theorem-consistency-svm}, Assumptions \ref{id-dep}, \ref{ld-mean}, \ref{aspn-suff-cond-consistency-domclass} are the most essential conditions for the slope consistency of the SVM estimator. Among the three conditions, Assumption \ref{aspn-suff-cond-consistency-domclass} is the only one that is not required for the slope consistency of QMLE for BCM in \cite{chang-park-yan-25}. While this condition is crucial for the slope consistency of SVM, which may fail when two classes are severely imbalanced, QMLE's slope consistency is not affected by class imbalance. 
\end{remark}

In what follows, we provide an example in which Assumptions \ref{id-dep}, \ref{ld-mean} are satisfied, so that any QMLE—which specifies that $U$ is independent of $X$ and has distribution function $G$ such that $\log G, \log (1-G)$ are strictly concave—is slope consistent. 
However, depending on parameter values, Assumption \ref{aspn-suff-cond-consistency-domclass} may fail under some specifications when, roughly speaking, two classes are severely imbalanced in the sense that $Y$ takes on one value much more frequently than another.  
Furthermore, we show by Monte Carlo simulations that the slope consistency of the SVM is significantly influenced by the validity of Assumption \ref{aspn-suff-cond-consistency-domclass}.

\begin{figure}[!h]   
\caption{Existence of $c_\ast, r_\ast$ under $V=_d \mathbb N(\mu,1), U=_d \mathbb N(0,1)$}
\label{fig-existence-c-r}
\begin{center}
     
\includegraphics[width=0.7\textwidth]{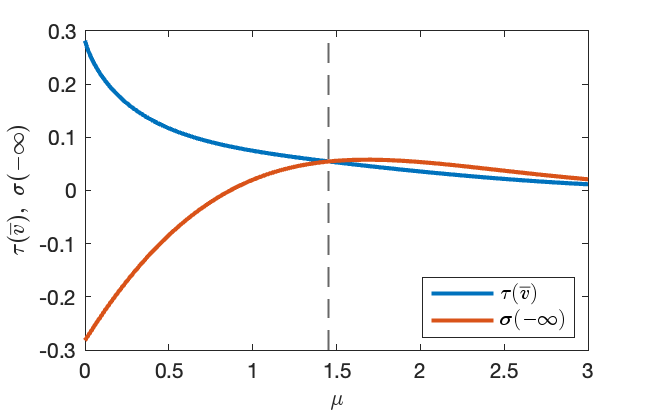}
\includegraphics[width=0.7\textwidth]{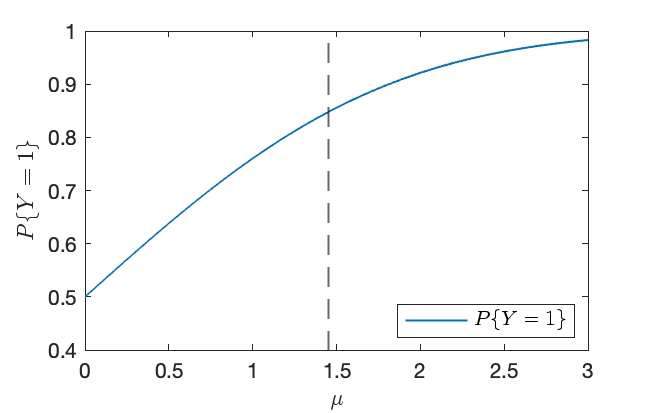}

\end{center}

\footnotesize
Notes: Under the specification $V=_d \mathbb N(\mu,1), U=_d \mathbb N(0,1)$ with varying $\mu$, the upper panel plots $\tau(\overline v), \sigma(-\infty)$ for each $\mu$.  
The broken line is the $\mu$ value at which $\tau(\overline v) =\sigma(-\infty)$.  
The lower graph plots $\mathbb P\{Y=1\}$ for each $\mu$ to indicate how imbalanced the two classes are. 
\end{figure}

\paragraph{Illustration: Imbalance and the SVM Slope Consistency} Let 
\[
V =_d \mathbb N(\mu,1), \quad \mbox{and} \quad U=_d \mathbb N(0,1)
\]
The underlying BCM can be that
\begin{equation}\label{dgp-illu-1}
Y =\sgn( \alpha+ \beta_1 X_1 + \beta_2 X_2 - U )
\end{equation}
where  $X_1,X_2,U =_d \mathbb N(0,1)$ are independent, and $\alpha=\mu, \beta_1 = \beta_2 = 1/\sqrt{2}$. 
Regardless of the value of $\mu$, any QMLE—which specifies that $U$ is independent of $X$ and has distribution function $G$ such that $\log G, \log (1-G)$ are strictly concave—is slope consistent, since Assumptions \ref{id-dep}, \ref{ld-mean} are satisfied; see Theorem 3.2 in \cite{chang-park-yan-25}.  
However, in what follows, we show that the SVM is consistent only when $\mu>0$ is not too large so that Assumption \ref{aspn-suff-cond-consistency-domclass} is satisfied. 

We vary $\mu$ and compute numerically $\tau(\overline{v}), \sigma(-\infty)$ 
for each $\mu$. The upper graph in Figure \ref{fig-existence-c-r} plots $\tau(\overline{v}), \sigma(-\infty)$ for each $\mu\in [0,3]$.  
When $\mu$ is small, $\tau(\overline{v}) >  \sigma(-\infty)$, i.e., Assumption \ref{aspn-suff-cond-consistency-domclass} is satisfied. However, when $\mu$ is greater than a certain value, Assumption \ref{aspn-suff-cond-consistency-domclass} fails, i.e., $\tau(\overline{v}) \leq  \sigma(-\infty)$. The broken line plots this threshold value, which is around 1.453. 
The lower graph in Figure \ref{fig-existence-c-r} plots $\mathbb P\{Y=1\}$ to indicate how imbalanced the two classes are.

We further show by Monte Carlo simulations that, the slope consistency of the SVM depends on whether Assumption \ref{aspn-suff-cond-consistency-domclass} holds or not.  
More specifically, for each $\alpha\in \{0,1,1.5,2,3\}$, we generate $n\in\{1000,2000,5000\}$ observations according to the DGP given by \eqref{dgp-illu-1} with $Nsim=2000$ replications. 
Note that Assumption \ref{aspn-suff-cond-consistency-domclass} holds for $\alpha\in\{0,1\}$, whereas it fails when $\alpha\in\{1.5,2,3\}$. 

\begin{table}[h!]
\caption{SVM slope (in-)consistency and degree of imbalance by simulation}
\label{table-svm-imbalance-wtsvm} 

\begin{center}
\centering
\begin{tabular}{c c c c c c c c c }
        \hline \hline					 
																					
		&		&	&	$\alpha=0$	&	$\alpha=1$	&	$\alpha=1.5$	&	$\alpha=2$	&	$\alpha=3$	\\ [0.5ex] \hline
		&		&	&		&		&		&		&		\\
		&	$\mathbb P\{Y=1\}$	&	&	0.5000	&	0.7602	&	0.8556	&	0.9214	&	0.9831	\\  
		&		&	&		&		&		&		&		\\ 
		\multicolumn{8}{c}{Panel A: SVM } 													\\ [1ex]
\multirow{2}{*}{1000 obs.}		&	Mean bias	&	&	0.0818	&	0.0951	&	0.5284	&	2.3083	&	6.0928	\\
		&	RMSE	&	&	0.1023	&	0.1215	&	3.0976	&	9.2884	&	58.4520	\\
		&		&	&		&		&		&		&		\\
\multirow{2}{*}{2000 obs.}		&	Mean bias	&	&	0.0585	&	0.0653	&	0.3412	&	4.5894	&	4.3255	\\
		&	RMSE	&	&	0.0729	&	0.0819	&	0.5149	&	66.3496	&	17.7460	\\
		&		&	&		&		&		&		&		\\
\multirow{2}{*}{5000 obs.}		&	Mean bias	&	&	0.0368	&	0.0403	&	0.2943	&	1.9163	&	4.8331	\\
		&	RMSE	&	&	0.0462	&	0.0503	&	0.4075	&	11.5083	&	27.3271	\\
		&		&	&		&		&		&		&		\\
		\multicolumn{8}{c}{Panel B: Class-weighted SVM } 													\\ [1ex]
\multirow{2}{*}{1000 obs.}		&	Mean bias	&	&	0.0818	&	0.0956	&	0.1105	&	0.1384	&	0.3251	\\
		&	RMSE	&	&	0.1022	&	0.1215	&	0.1397	&	0.1764	&	0.5934	\\
		&		&	&		&		&		&		&		\\
\multirow{2}{*}{2000 obs.}		&	Mean bias	&	&	0.0588	&	0.0639	&	0.0761	&	0.0983	&	0.1955	\\
		&	RMSE	&	&	0.0732	&	0.0810	&	0.0959	&	0.1242	&	0.2590	\\
		&		&	&		&		&		&		&		\\
\multirow{2}{*}{5000 obs.}		&	Mean bias	&	&	0.0369	&	0.0417	&	0.0488	&	0.0621	&	0.1181	\\
		&	RMSE	&	&	0.0464	&	0.0521	&	0.0614	&	0.0773	&	0.1499	\\
        
        \hline \hline 
\end{tabular}
\end{center}

\footnotesize
Notes: The DGP is $Y = \sgn(\alpha+\beta_1X_1+\beta_2X_2 - U)$, where $\beta_1=\beta_2=1/\sqrt{2}$, and  $X_1, X_2, U$ are independent. We set $X_1, X_2, U=_d \mathbb{N}(0,1)'$, so $V =_d \mathbb N(\alpha,1)$.  We select $\alpha\in\{0,1,1.5,2,3\}$, and note that Assumption \ref{aspn-suff-cond-consistency-domclass} holds for $\alpha\in\{0,1\}$, whereas it fails when $\alpha\in\{1.5,2,3\}$. 
The Monte Carlo simulations have $Nsim=2000$ replications, and for each replication, we choose sample sizes $n\in\{1000, 2000, 5000\}$. 
This table presents the bias and root mean square error of the rescaled slope estimator $\tilde \beta :=\hat \beta_1/\hat \beta_2$ given by SVM (Panel A) and class-weighted SVM (Panel B). 
\end{table}

Panel A in Table~\ref{table-svm-imbalance-wtsvm} presents the bias and root mean square error of the rescaled slope estimator $\tilde \beta := \hat \beta_1/\hat \beta_2$ given by the SVM, along with $\mathbb P \{Y=1\}$ for each $\alpha$. Panel A in Table~\ref{table-svm-imbalance-wtsvm} shows that, when $\alpha\in\{1.5,2,3\}$, the SVM slope is inconsistent since Assumption \ref{aspn-suff-cond-consistency-domclass}  fails, whereas it is consistent when $\alpha\in\{0,1\}$ and Assumption \ref{aspn-suff-cond-consistency-domclass} holds.

\subsection{Addressing Inconsistency of SVM due to Imbalance Using Class-weighted SVM}\label{section-svm-class-wgt}

In practice, applying the SVM is often coupled with the use of class weights, especially when there are more observations with $Y$ taking one value than another. More specifically, one may assign higher weights to the observations in the minority class and lower weights to those in the majority class, so that the number of observations in the two classes appears balanced. 
In this section, we show that the class-weighted SVM is slope consistent under Assumptions \ref{id-dep} and \ref{ld-mean}, 
without requiring the non-severe imbalance condition in Assumption \ref{aspn-suff-cond-consistency-domclass}.

A class-weighted SVM corresponds to minimizing
\[
Q_n^{w}(\theta) = \frac{1}{n} \sum_{i=1}^{n} \left[ 1\{Y_i=1\} \big( 1-(\alpha+X_i'\beta) \big)_{+} +  \hat w 1\{Y_i=-1\} \big( 1+(\alpha+X_i'\beta) \big)_{+} \right] + \frac{\lambda_n}{n}\|\beta\|^2
\]
where $\hat w$ is a weight assigned to the observations in the class with $Y=-1$. Let the class-weighted SVM estimator $\hat \theta^{w} =(\hat \alpha^{w}, \hat \beta^{w})'$ be a minimum of $Q_n^{w}(\theta)$. Define
\[
Q^{w}(\theta) = \mathbb E \left[ 1\{Y=1\} \big( 1-(\alpha+X'\beta) \big)_{+} +  w 1\{Y=-1\} \big( 1+(\alpha+X'\beta) \big)_{+} \right].
\]
where $w = \text{plim } \hat w$. 

In what follows, we set $\hat w = \hat{\mathbb P}\{Y=1\} /\hat{\mathbb P}\{Y=-1\}$, the ratio of sample proportions of observations coupled with $Y=1$ and $Y=-1$, since this is the most commonly used choice in practice. Thus, 
\[
w = \frac{\mathbb{P}\{Y=1\} }{\mathbb{P}\{Y=-1\} }. 
\]

\begin{theorem}\label{theorem-consistency-class-wgt-svm}
Let $(Y,X)$ be generated by the BCM satisfying Assumptions \ref{aspn-med0-error}, \ref{aspn-pos-density-Xm}, \ref{aspn-svm-hessian}, \ref{id-dep}, \ref{ld-mean}. Then the class-weighted SVM slope estimator $\hat\beta_w$ is consistent for $\beta_0$ (up to multiplication by a positive scalar).
\end{theorem}

Theorem \ref{theorem-consistency-class-wgt-svm} shows that the class-weighted SVM effectively addresses the inconsistency of the SVM that arises purely due to severe imbalance, i.e., when $\mathcal L(U|X) = \mathcal L(U|V)$ and $\mathbb E(X|V) = aV+ b$ are satisfied, but the non-severe imbalance condition in Assumption \ref{aspn-suff-cond-consistency-domclass} fails.

Panel B in Table \ref{table-svm-imbalance-wtsvm} demonstrates the consistency of the class-weighted SVM. In particular, when $\alpha\in \{1.5,2,3\}$ and Assumption \ref{aspn-suff-cond-consistency-domclass} is violated, the class-weighted SVM is still consistent, whereas the SVM is inconsistent as shown in Panel A.

\subsection{Finite Sample Performances}\label{section-illustration}

In this section, we illustrate by Monte Carlo experiments that, when both the SVM and logistic regression (as a QMLE) are slope consistent, their finite-sample performance may differ across specifications, with one potentially outperforming the other.

\begin{table}[t!]
\caption{Comparison of $\hat\beta_{SVM}$ and $\hat \beta_{LR}$ by simulation}
\label{table-svm-lr} 
\begin{center}
\centering
\begin{tabular}{c c c c c c c c c }
        \hline \hline					
        &		&		&	&	$\mu=0$	&	$\mu=1$	&	$\mu=2$	&	$\mu=3$	&	$\mu=4$	\\ [0.5ex] \hline
	&		&		&	&		&		&		&		&		\\
\multirow{4}{*}{250 obs.}	&	\multirow{2}{*}{Mean bias}	&	$\tilde \beta_{SVM}$	&	&	0.0844	&	0.0833	&	0.0706	&	0.0501	&	0.0381	\\
	&		&	$\tilde \beta_{LR}$	&	&	0.0755	&	0.0749	&	0.0778	&	0.0798	&	0.0805	\\ [1ex]
	&	\multirow{2}{*}{RMSE}	&	$\tilde \beta_{SVM}$	&	&	0.1058	&	0.1047	&	0.0902	&	0.0653	&	0.0497	\\
	&		&	$\tilde \beta_{LR}$	&	&	0.0945	&	0.0945	&	0.0984	&	0.1010	&	0.1012	\\
	&		&		&	&		&		&		&		&		\\
\multirow{4}{*}{500 obs.}	&	\multirow{2}{*}{Mean bias}	&	$\tilde \beta_{SVM}$	&	&	0.0591	&	0.0586	&	0.0485	&	0.0337	&	0.0257	\\
	&		&	$\tilde \beta_{LR}$	&	&	0.0530	&	0.0535	&	0.0553	&	0.0552	&	0.0553	\\  [1ex]
	&	\multirow{2}{*}{RMSE}	&	$\tilde \beta_{SVM}$	&	&	0.0735	&	0.0739	&	0.0612	&	0.0426	&	0.0324	\\
	&		&	$\tilde \beta_{LR}$	&	&	0.0665	&	0.0672	&	0.0695	&	0.0694	&	0.0696	\\
	&		&		&	&		&		&		&		&		\\
\multirow{4}{*}{2000 obs.}	&	\multirow{2}{*}{Mean bias}	&	$\tilde \beta_{SVM}$	&	&	0.0300	&	0.0296	&	0.0239	&	0.0168	&	0.0129	\\
	&		&	$\tilde \beta_{LR}$	&	&	0.0271	&	0.0268	&	0.0278	&	0.0285	&	0.0287	\\  [1ex]
	&	\multirow{2}{*}{RMSE}	&	$\tilde \beta_{SVM}$	&	&	0.0375	&	0.0367	&	0.0298	&	0.0209	&	0.0161	\\
	&		&	$\tilde \beta_{LR}$	&	&	0.0338	&	0.0335	&	0.0349	&	0.0357	&	0.0361	\\
        \hline \hline 
\end{tabular}
\end{center}

\footnotesize
Notes: The DGP is $Y = \sgn(\alpha+\beta_1X_1+\beta_2X_2 - U)$, where $\alpha=1/2, \beta_1=0, \beta_2=1$, and  $X_1, X_2, U$ are independent. We set $X_1, U=_d \mathbb{N}(0,1)'$, and $X_2$ is specified as the normal mixture $\frac{1}{2}\mathbb N(-\mu,1) + \frac{1}{2} \mathbb N(\mu,1)$ standardized to have variance one, where we select $\mu\in\{0,1,2,3,4\}$. 
The Monte Carlo simulations have $Nsim=2000$ replications, and for each replication, we choose sample sizes $n\in\{250, 500, 2000\}$. 
This table compares the bias and root mean square error of the slope estimator $\hat \beta$ given by the SVM and logistic regression (LR). 
For a fair comparison, we consider the rescaled slope estimator $ \tilde \beta = \hat \beta_1/\hat\beta_{2}$.
\end{table}

For the data-generating processes, we assume that 
\[
Y = \mbox{sgn}( \alpha+ \beta_1 X_1 + \beta_2 X_2 - U)
\]
where $\beta_1 = 0, \beta_2 =1$, $\alpha=1/2$. 
We let $X_1, U =_d \mathbb{N}(0,1)$, and $X_1,X_2,U$ be independent of each other. 
The distribution of $X_2$ is specified based on the following normal mixture distribution standardized to have variance one, 
\begin{align*}
	X_2 =_d \mbox{std} \left( \frac{1}{2} \mathbb{N}(-\mu,1) +  \frac{1}{2} \mathbb{N}(\mu,1)  \right)
\end{align*}
for some $\mu$, where $\mbox{std}(\cdot)$ denotes standardization to mean zero and variance one. Note that Assumption \ref{ld-mean} is satisfied since $V=\alpha+X_2$ and $\mathbb E(X_1|V) = 0$. Assumption \ref{aspn-suff-cond-consistency-domclass} is verified numerically for $\mu\in\{0,1,2,3,4\}$, and thus, the SVM slope estimator is consistent.

For each $\mu\in\{0,1,2,3,4\}$, we generate $n\in\{250,500,2000\}$ observations with $Nsim=2000$ replications. In each replication, we obtain the estimates $(\hat \alpha, \hat \beta_1, \hat \beta_2)'$ of the parameter $(\alpha,\beta_1,\beta_2)$ given by the SVM and QMLE specifying the logistic distribution for the error term, referred to as logistic regression (LR). 
To compare the finite sample performance of estimators given by the SVM and LR, we rescale the estimators so that the last component of the rescaled estimator is set to 1. We divide each of the estimators by its last component, 
and, as before, focus on the rescaled slope estimator $\widetilde{\beta} = \hat \beta_1/\hat \beta_2$. 
Table \ref{table-svm-lr} presents the mean bias and RMSE of the rescaled slope estimator given by the SVM and LR. 
This table shows that, when $\mu\in\{0,1\}$, the LR slope estimator is more precise than that of the SVM, whereas the SVM performs better when $\mu\in\{2,3,4\}$.

\section{Conclusion}

In this paper, we show that the slope of the separating hyperplane given by the SVM consistently estimates the BCM slope parameter under the linear conditional mean condition for covariates given the systematic component used in the QMLE slope consistency literature. The SVM slope estimator is asymptotically equivalent to that of logistic regression in this sense. When binary outcomes are not severely imbalanced, this asymptotic equivalence holds as stated. Otherwise, the SVM should be used with class weights for this asymptotic equivalence to hold. We establish consistency of the SVM estimator for the slope parameter without assuming that a pseudo-true value exists. This makes our proof nonstandard and nontrivial. It is, however, essential because a pseudo-true value may not exist for the SVM estimator when there is severe imbalance in binary outcomes.

The asymptotic equivalence of the SVM estimator with the QMLE is quite surprising, since the SVM was originally developed as an algorithm that does not rely on any statistical or econometric model. Their asymptotics are indeed largely comparable, although the SVM estimator cannot be defined as a QMLE. The conditions required for slope consistency of the SVM estimator and QMLE are not overly restrictive, and can be made to hold by giving appropriate weights to the observations. For the BCM, once a consistent estimator for the slope parameter is available, we may readily obtain a consistent estimator for the intercept parameter as well. The asymptotic equivalence we established in the paper relates the SVM and logistic regression, the two most commonly used machine learning tools, more closely to each other. However, their finite-sample performances can be quite distinct depending on the distributions of covariates and errors. It seems that they will both survive, since neither dominates the other.

\clearpage 

\section*{Appendix}

\appendix

\section{Proof of Lemma~\ref{lemma-gradient-hessian-svm}}

We introduce some notations for the proof. Let $X_{-m}= (X_1,X_2,\dots,X_{m-1})'$, and $\tilde X = (1,X)$. We write $\tilde x_{j}$ the $j$-th element of $\tilde x$. 
Recall $\Pi(x) = \mathbb P\{Y=1|X=x\}$. 
When the arguments of $\Pi(x)$ need to be more explicit, let $\Pi(x_m,x_{-m} ) = \mathbb P\{ Y=1|X_m=x_m, X_{-m}=x_{-m} \}$. 
Let $\iota_j \in \mathbb R^{1+m}$ be the vector whose $j$-th element is 1 and all other elements are zero. 
Let $a \wedge b = \min\{a,b\}$ for any $a,b\in \mathbb R$. 

In what follows, we let $\theta =(\alpha,\beta')'$ be fixed arbitrarily such that $\beta_m\neq 0$. By Assumption \ref{aspn-pos-density-Xm} (a) and $\beta_m \neq 0$, $\theta'\tilde X$ has a Lebesgue density.  

\paragraph{Gradient:} We first show that \eqref{gradient-svm} is the gradient of $Q(\theta)$. 
Let $j=1,\dots,1+m$ be fixed arbitrarily. 
It suffices to show 
\begin{align}\label{pf-gradient-Q-main}
\frac{\partial}{\partial \theta_{j}} Q(\theta) 
= - \mathbb E   \Pi(X) 1\big\{  \theta' \tilde X < 1 \big\} \tilde X_j  
 + \mathbb E (1-\Pi(X)) 1\big\{ \theta' \tilde X > -1 \big\} \tilde X_j  
\end{align}
Write 
\begin{align*}
Q(\theta)  
& = \mathbb E 1\{Y=1\} ( 1- (\alpha+X'\beta))_{+}  + \mathbb E 1\{Y=-1\} ( 1+(\alpha+X'\beta))_{+} \\
& = \mathbb E \Pi(X) (1- \theta'\tilde X)_{+} + \mathbb E (1-\Pi(X) ) (1 + \theta'\tilde X)_{+} \\
& =: Q_{+}(\theta) + Q_{-}(\theta) . 
\end{align*} 
Note that 
\begin{align*}
&  Q_{+}(\theta+ t \iota_j) - Q_{+}(\theta)  
= \mathbb E \Pi(X) \big( (1- \theta'\tilde X - t \tilde X_j)_{+}  - (1- \theta'\tilde X )_{+} \big)  \notag \\
& = \mathbb E \Pi(X) 1\big\{ \theta' \tilde X \leq (1-t\tilde X_j)\wedge 1  \big\} (-t\tilde X_j)   \notag \\
& \quad + \mathbb E \Pi(X)  1\big\{t\tilde X_j >0, 1-t\tilde X_j < \theta' \tilde X \leq  1\big\}  \big(- (1-\theta'\tilde X) \big) \notag  \\
& \quad + \mathbb E \Pi(X)  1\big\{t\tilde X_j \leq 0, 1 < \theta' \tilde X \leq 1-t\tilde X_j \big\}  \big( 1-\theta'\tilde X -t\tilde X_j \big)  \label{pf-gradient-Qpos-aux1} 
\end{align*}
By the dominated convergence theorem and that $\theta'\tilde X$ has a Lebesgue density,  
\begin{equation}\label{pf-gradient-Qpos-aux1}
\lim_{t\to 0} \frac{1}{t} \mathbb E \Pi(X) 1\big\{  \theta' \tilde X \leq (1-t\tilde X_j)\wedge 1 \big\} (-t\tilde X_j)  
= - \mathbb E \Pi(X) 1\big\{ \theta' \tilde X \leq 1 \big\} \tilde X_j  .
\end{equation}
Moreover, 
\begin{align}
& \lim_{t\to 0}  \bigg| \frac{1}{t}   \mathbb E \Big( \Pi(X)  1\big\{ t\tilde X_j >0, 1-t\tilde X_j < \theta' \tilde X \leq 1\big\}  (1-\theta'\tilde X )   \Big) \bigg|  \notag \\
& \leq \lim_{t\to 0}   \frac{1}{|t|} \mathbb E \Big(  \Pi(X) 1\big\{ 1-t\tilde X_j < \theta' \tilde X \leq 1\big\}  \big|t\tilde X_j\big|  \Big) \notag \\ 
& = \lim_{t\to 0} \mathbb E \Big(  \Pi(X) 1\big\{ 1-t\tilde X_j <\theta' \tilde X \leq 1\big\} \big|\tilde X_j\big|  \Big)  
= 0 , \label{pf-gradient-Qpos-aux2}
\end{align} 
where the last equality holds since $\theta'\tilde X$ has a Lebesgue density. 
Similarly, 
\begin{equation}\label{pf-gradient-Qpos-aux3}
\lim_{t\to 0}  \bigg| \frac{1}{t}   \mathbb E \Big( \Pi(X)   1\big\{ t\tilde X_j \leq 0, 1 < \theta' \tilde X \leq 1-t\tilde X_j  \big\}  (1-\theta'\tilde X - t\tilde X_j)   \Big)  \bigg| 
= 0 , 
\end{equation} 
upon noticing that $ t\tilde x_j \leq 0$ and $1 < \theta' \tilde x \leq 1-t\tilde x_j $  imply $ 0\leq 1-\theta'\tilde x - t\tilde x_j < -t \tilde x_j$. 
Combining \eqref{pf-gradient-Qpos-aux1}, \eqref{pf-gradient-Qpos-aux2}, \eqref{pf-gradient-Qpos-aux3} yields that 
\begin{align}\label{pf-gradient-Qpos-final}
& \frac{\partial}{\partial \theta_{j}} Q_{+}(\theta) 
= \lim_{t\to 0} \frac{1}{t} \Big( Q_{+}(\theta+ t \iota_j) - Q_{+}(\theta)  \Big)  = - \mathbb E \Pi(X) 1\big\{ \theta' \tilde X \leq 1 \big\} \tilde X_j  .
\end{align}
By virtually the same arguments, 
\begin{align*}
\frac{\partial}{\partial \theta_{j}} Q_{-}(\theta) = \mathbb E (1-\Pi(X)) 1\big\{  \theta' \tilde X \geq -1 \big\} \tilde X_j  . 
\end{align*}
Together with \eqref{pf-gradient-Qpos-final} and $\theta'\tilde X$ having a Lebesgue density, \eqref{pf-gradient-Q-main} holds as was to be shown.

\paragraph{Hessian and Continuity: }
Now we show that \eqref{hessian-svm} is the Hessian matrix of $Q(\theta)$. 
Let $j, k=1,\dots, m+1$ be fixed arbitrarily. 
By \eqref{pf-gradient-Q-main},
\begin{align}\label{pf-hessian-alpha-aux}
& \frac{\partial}{\partial \theta_k} \frac{\partial}{\partial \theta_j} Q(\theta) 
= \lim_{t\to 0} \frac{1}{t} \bigg( \frac{\partial}{\partial \theta_j} Q(\theta+t \iota_k ) - \frac{\partial}{\partial \theta_j} Q(\theta) \bigg) \notag  \\
& = \lim_{t\to 0} \frac{1}{t}  \bigg( \mathbb E h_+(t; X_{-m})  +   \mathbb E h_-(t; X_{-m})  \bigg), 
\end{align}
where
\begin{align*}
h_+(t; x_{-m}) 
& := \mathbb E \Big(  \Pi(X) \big( 1\{ \theta'\tilde X <1 \} -  1\{  \theta'\tilde X + t\tilde X_k <1 \} \big) \tilde X_j \Big| X_{-m} = x_{-m} \Big) \\
h_-(t; x_{-m}) 
& := \mathbb E \Big(  (1-\Pi(X)) \big( 1\{ \theta'\tilde X + t\tilde X_k > -1  \} -  1\{  \theta'\tilde X  >-1\} \big) \tilde X_j \Big| X_{-m} = x_{-m} \Big) .
\end{align*} 
Define $x_{m,\ast}^+(x_{-m}) = \frac{1-\alpha-\beta_{-m}'x_{-m}}{\beta_m}$, $x_{m,\ast}^-(x_{-m}) = \frac{-1-\alpha-\beta_{-m}'x_{-m}}{\beta_m}$, and
\begin{align*}
&\gamma^+(x_m,x_{-m}) =  p(x_m|x_{-m}) \Pi(x_m,x_{-m} ), \quad  \gamma^+_\ast(x_{-m}) = \gamma^+( x_{m,\ast}^+(x_{-m}), x_{-m} )   \\  
&\gamma^-(x_m,x_{-m}) =  p(x_m|x_{-m}) \big( 1-\Pi(x_m,x_{-m} ) \big), \quad \gamma^-_\ast(x_{-m}) =  \gamma^-( x_{m,\ast}^-(x_{-m}), x_{-m} ) .  
\end{align*} 
By Assumption \ref{aspn-svm-hessian} (b)-(c),  $\gamma^+(x_m,x_{-m})$, $\gamma^-(x_m,x_{-m})$ are continuous in $x_m$ and bounded by $K$.  
In what follows, we focus on the limit in \eqref{pf-hessian-alpha-aux} when $t\to 0+$, since the proof for the limit when $t\to 0-$ is essentially the same. 

Consider the first case that $k\neq m+1$. For each $t,x_{-m}$, define the interval 
\[
I(t, x_{-m} ) =  
\Big\{ x_m \in \mathbb R:  x_m \text{ between } x_{m,\ast}^+(x_{-m}) - t\tilde x_k/\beta_m  \text{ and }  x_{m,\ast}^+(x_{-m})  
\Big\}, 
\]
which has length $\frac{ |t\tilde x_k| }{ |\beta_m|}$. 
Then 
\begin{align}\label{pf-hessian-pos-hpos}
h_+(t; x_{-m})  
& =  
1\{ t\tilde x_k >0\} \int  \gamma^+(x_m,x_{-m} )  1\big\{ x_m \in I(t, x_{-m}) \big\}  \tilde x_j dx_m \notag  \\
& \quad - 1\{ t\tilde x_k <0\} \int  \gamma^+(x_m,x_{-m} )  1\big\{ x_m \in I(t, x_{-m}) \big\}  \tilde x_j dx_m . 
\end{align}
Considering cases $\tilde x_k \gtreqless 0$ separately in \eqref{pf-hessian-pos-hpos} yields 
\begin{align}\label{pf-hessian-pos-aux-1}
\lim_{t\to 0+} \frac{h_+(t; x_{-m}) }{t}  
= \begin{cases}
\gamma^+_\ast(x_{-m}) \frac{  \tilde x_k }{ |\beta_m|} x_{m,\ast}^+(x_{-m}) &  \text{ if } j=m+1  \\
\gamma^+_\ast(x_{-m}) \frac{  \tilde x_k }{ |\beta_m|} \tilde x_j & \text{ if } j\neq m+1
\end{cases}, 
\end{align}
since $\gamma^+(x_m,x_{-m})$, $\gamma^-(x_m,x_{-m})$ are continuous in $x_m$. 
Thus, 
\begin{align}
&  
\lim_{t\to 0+} \mathbb E \frac{h_+(t; X_{-m})}{t}  
= \mathbb E  \lim_{t\to 0+} \frac{h_+(t; X_{-m})}{t}  \notag \\
& = 
\begin{cases}
\frac{1}{|\beta_m|} \mathbb E  \gamma^+_\ast(X_{-m})  \tilde X_k  x_{m,\ast}^+(X_{-m}) &  \text{ if } j=m+1  \\
\frac{1}{|\beta_m|} \mathbb E  \gamma^+_\ast(X_{-m})  \tilde X_k  \tilde X_j & \text{ if } j\neq m+1
\end{cases}  \label{pf-hessian-pos-aux-2}  \\
& =  \mathbb E \Pi(X) \delta\big( 1- (\alpha+X'\beta) \big) \tilde X_j \tilde X_k  \label{pf-hessian-pos-aux-3}
\end{align}
where the first equality is by the dominated convergence theorem and $\big|\gamma^+(x_m,x_{-m}) \big| \leq K$. 
The second equality in \eqref{pf-hessian-pos-aux-2} is by \eqref{pf-hessian-pos-aux-1}. 
To see the third equality in \eqref{pf-hessian-pos-aux-3}, notice 
\begin{align*}
&   \mathbb E \Big( \Pi(X) \delta\big( 1- (\alpha+X'\beta) \big) \tilde X_j \tilde X_k \Big| X_{-m} = x_{-m} \Big) \\
& = \int \gamma^+(x_m,x_{-m} ) \delta\big( 1- (\alpha+\beta_{-m}'x_{-m} + \beta_m x_m ) \big)  \tilde x_j \tilde x_k dx_{m} \\
& = 
\begin{cases}
\frac{\tilde x_k}{|\beta_m| } \gamma^+_\ast(x_{-m} )  x_{m,\ast}^+(x_{-m})  &  \text{ if } j=m+1 \\
\frac{\tilde x_k}{|\beta_m| } \gamma^+_\ast(x_{-m} ) \tilde x_j &  \text{ if } j\neq m+1 ,
\end{cases}
\end{align*}
where the last line follows from the property of the Dirac delta function that 
\begin{align}\label{dirac-delta-property}
\int \delta(c_1 t + c_2) f(t) dt = \frac{1}{|c_1|} f(-c_2/c_1), \quad \mbox{for any} \quad c_1\neq 0, c_1,c_2 \in \mathbb R . 
\end{align}
Using the same arguments, we can show
\begin{align*}
& \lim_{t\to 0+} \mathbb E \frac{h_{-}(t; X_{-m})}{t}   
= 
\begin{cases}
\frac{1}{|\beta_m|} \mathbb E  \gamma^-_\ast(X_{-m})  \tilde X_k x_{m,\ast}^-(X_{-m}) &  \text{ if } j=m+1  \\
\frac{1}{|\beta_m|} \mathbb E  \gamma^-_\ast(X_{-m})  \tilde X_k  \tilde X_j & \text{ if } j\neq m+1
\end{cases}  \\
& =  \mathbb E (1-\Pi(X)) \delta\big( 1+ (\alpha+X'\beta) \big) \tilde X_j \tilde X_k  .
\end{align*}
Together with \eqref{pf-hessian-pos-aux-2} and \eqref{pf-hessian-pos-aux-3}, \eqref{pf-hessian-alpha-aux} implies 
\begin{align}
& \frac{\partial}{\partial \theta_k} \frac{\partial}{\partial \theta_j} Q(\theta)  \notag \\
& = \begin{cases}
\frac{1}{|\beta_m|} \mathbb E \Big(  \gamma^+_\ast(X_{-m}) x_{m,\ast}^+(X_{-m}) +   \gamma^-_\ast(X_{-m}) x_{m,\ast}^-(X_{-m}) \Big) \tilde X_k    &  \text{ if } j=m+1  \\
\frac{1}{|\beta_m|} \mathbb E \Big(  \gamma^+_\ast(X_{-m}) +   \gamma^-_\ast(X_{-m}) \Big) \tilde X_k  \tilde X_j & \text{ if } j\neq m+1
\end{cases} \label{pf-hessian-case1-main}  \\
& =  \mathbb E \Big( \Pi(X)) \delta\big( 1- (\alpha+X'\beta) \big) +  (1-\Pi(X)) \delta\big( 1+ (\alpha+X'\beta) \big) \Big) \tilde X_j \tilde X_k   \notag \\
& = \mathbb E \delta\big( 1 - Y (\alpha+X'\beta) \big) \tilde X_j \tilde X_k , \label{pf-hessian-case1-main-2}
\end{align}
which shows \eqref{hessian-svm} for case $j, k =1,\dots, m+1$ such that $k\neq m+1$.

\medskip
For the other case that $k = m+1$, define $x_{m}^+(t,x_{-m}) = \frac{1- \alpha- \beta_{-m}'x_{-m}}{\beta_m + t} $ for each $x_{-m}$ and all $t>0$ small. 
Then define the interval
\[
J(t, x_{-m} ) =  
\big\{ x_m \in \mathbb R:  x_m \text{ between } x_{m}^+(t,x_{-m}) \text{ and }  x_{m,\ast}^+(x_{-m}) 
\big\}. 
\]
Note that
$J(t,x_{-m})$ has length $\frac{  |t| }{ | \beta_m+t|} |x_{m}^+(t,x_{-m})| $. 
Considering cases $\beta_m \gtrless 0$ separately for $h_+(t; x_{-m}) $ yields 
\begin{align*} 
h_+(t; x_{-m})  
& = \sgn(\beta_m) 1\big\{  1-\alpha-\beta_{-m}'x_{-m} >0 \big\}   \int_{J(t, x_{-m})}  \gamma^+(x_m,x_{-m} ) \tilde x_j  dx_m \notag  \\
& \quad -  \sgn(\beta_m)  1\big\{  1-\alpha-\beta_{-m}'x_{-m} <0 \big\}    \int_{J(t, x_{-m})}  \gamma^+(x_m,x_{-m} )  \tilde x_j dx_m, 
\end{align*}
and thus, 
\begin{align*}
\lim_{t\to 0+} \frac{h_+(t; x_{-m}) }{t}   
= 
\begin{cases} 
\frac{1}{|\beta_m|}  \gamma^+_\ast(x_{-m})  x_{m,\ast}^+ (x_{-m})  x_{m,\ast}^+ (x_{-m})   & \text{ if } j= m+1 \\
\frac{1}{|\beta_m|} \gamma^+_\ast(x_{-m}) x_{m,\ast}^+ (x_{-m}) \tilde x_j  & \text{ if } j\neq m+1 .
\end{cases}
\end{align*}
Using the same arguments as for the previous case $k\neq m+1$, 
\begin{align*}
&  \lim_{t\to 0+} \mathbb E \frac{h_+(t; X_{-m})}{t}  
= 
\begin{cases}
\frac{1}{|\beta_m|} \mathbb E  \gamma^+_\ast(X_{-m})  x_{m,\ast}^+ (X_{-m}) x_{m,\ast}^+ (X_{-m})&  \text{ if } j=m+1  \\
\frac{1}{|\beta_m|} \mathbb E  \gamma^+_\ast(X_{-m})  x_{m,\ast}^+ (X_{-m})  \tilde X_j & \text{ if } j\neq m+1
\end{cases}    \\
& =  \mathbb E \Pi(X) \delta\big( 1- (\alpha+X'\beta) \big) \tilde X_j \tilde X_k . 
\end{align*}
Similarly, 
\begin{align*}
&  \lim_{t\to 0+} \mathbb E \frac{h_-(t; X_{-m})}{t}  
= 
\begin{cases}
\frac{1}{|\beta_m|} \mathbb E  \gamma^-_\ast(X_{-m}) x_{m,\ast}^- (X_{-m})  x_{m,\ast}^- (X_{-m}) &  \text{ if } j=m+1  \\
\frac{1}{|\beta_m|} \mathbb E  \gamma^-_\ast(X_{-m}) x_{m,\ast}^- (X_{-m})  \tilde X_j & \text{ if } j\neq m+1
\end{cases}    \\
& =  \mathbb E ( 1- \Pi(X) ) \delta\big( 1+ (\alpha+X'\beta) \big) \tilde X_j \tilde X_k . 
\end{align*}
Therefore, for $k=1+m$, 
\begin{align}
& \frac{\partial}{\partial \theta_{k}} \frac{\partial}{\partial \theta_j} Q(\theta) \notag   \\
& = \begin{cases}
\frac{1}{|\beta_m|} \mathbb E \Big(  \gamma^+_\ast(X_{-m}) \big( x_{m,\ast}^+ (X_{-m})  \big)^2 +   \gamma^-_\ast(X_{-m}) \big( x_{m,\ast}^- (X_{-m})  \big)^2  \Big)    &  \text{if } j=m+1   \\
\frac{1}{|\beta_m|} \mathbb E \Big(  \gamma^+_\ast(X_{-m}) x_{m,\ast}^+ (X_{-m}) +   \gamma^-_\ast(X_{-m}) x_{m,\ast}^- (X_{-m})  \Big)   \tilde X_j & \text{if } j\neq m+1 
\end{cases}  \label{pf-hessian-case2-main} \\
& =  \mathbb E \Big( \Pi(X)) \delta\big( 1-(\alpha+X'\beta) \big) +  (1-\Pi(X)) \delta\big( 1+ (\alpha+X'\beta) \big) \Big) \tilde X_j \tilde X_k  \notag \\
& = \mathbb E \delta\big( 1 - Y (\alpha+X'\beta) \big) \tilde X_j \tilde X_{k} . \label{pf-hessian-case2-main-2}
\end{align}
Therefore, the proof for \eqref{hessian-svm} is complete by \eqref{pf-hessian-case1-main-2} for $j,k=1,\dots,m+1$ such that $k\neq m+1$, and \eqref{pf-hessian-case2-main-2} for case $k=m+1$. 

Furthermore, \eqref{pf-hessian-case1-main} and  \eqref{pf-hessian-case2-main} imply that $\frac{\partial}{\partial \theta_k} \frac{\partial}{\partial \theta_j} Q(\theta)$ is continuous in $\theta$, due to Assumption~\ref{aspn-svm-hessian} (b)-(c).

\paragraph{Positive Definiteness of Hessian}  
To show $\ddot{Q}(\theta)$ is positive definite, it suffices to show that for any $\theta_\bullet = (\alpha_\bullet, \beta_\bullet')' \neq 0$,  $\theta_\bullet'\ddot{Q}(\theta)\theta_\bullet > 0$. 
Notice 
\begin{align*}
\theta_\bullet'\ddot{Q}(\theta)\theta_\bullet = R_1 + R_2,
\end{align*}
where
\begin{align*}
R_1 & := f_{\alpha+X'\beta}(+1) \mathbb{P}\{ Y = +1 |\alpha + X'\beta = +1 \} \mathbb{E} \big( (\theta_\bullet' \tilde X)^2 \big| Y = +1, \alpha + X'\beta = +1   \big) \\
R_2 & :=  f_{\alpha+X'\beta}(-1) \mathbb{P}\{ Y = -1 |\alpha + X'\beta = -1 \} \mathbb{E} \big(  (\theta_\bullet' \tilde X)^2 \big| Y = -1, \alpha + X'\beta = -1   \big) .
\end{align*}
Since $\alpha+X'\beta$  has a positive density on $\mathbb{R}$, $f_{\alpha+X'\beta}(\pm 1)>0$.  
Moreover, 
\[
\mathbb{P}\{ Y = +1 |\alpha + X'\beta = +1 \}, \quad  \mathbb{P}\{ Y = -1 |\alpha + X'\beta = -1 \} >0.
\]
by Assumption \ref{aspn-pos-density-Xm} (b). 
Thus, to show $\theta_\bullet'\ddot{Q}(\theta)\theta_\bullet = R_1 + R_2>0$, it suffices to show that at least one of the following holds, i.e.,  
\begin{align*}
& \mathbb{E} \big( (\theta_\bullet' \tilde X)^2 \big| Y = +1, \alpha + X'\beta = +1   \big) >0 \\
& \mathbb{E} \big( (\theta_\bullet' \tilde X)^2 \big| Y = -1, \alpha + X'\beta = -1   \big) >0, 
\end{align*}
which is satisfied if at least one of the following holds, i.e.,  
\begin{align*}
S_1 &:= \mathbb P\big\{ \theta_\bullet' \tilde X = 0  \big| Y = +1, \alpha + X'\beta = +1  \big\} < 1 \\
S_2 &:= \mathbb P\big\{ \theta_\bullet' \tilde X = 0  \big| Y = -1, \alpha + X'\beta = -1  \big\} < 1 .
\end{align*}
We will verify it for both cases of $\theta_\bullet$:  (I) $\beta_\bullet \in \text{span}(\beta)$, (II) $\beta_\bullet \notin \text{span}(\beta)$.  

For case (I) that $\beta_\bullet \in \text{span}(\beta)$, $\beta_\bullet = c_1\beta$ for some $c_1\in \mathbb R$. 
Then $\theta_\bullet' \tilde X = c_1(\alpha+\beta'X) +c_2$ for $c_2:= \alpha_\bullet - c_1 \alpha \in \mathbb R$. 
Thus,  
\[
S_1 
= \mathbb P\{ c_1+c_2 = 0\big| Y = +1, \alpha + X'\beta = +1 \} ,
\]
which equals one if and only if $c_1=-c_2$, 
and 
\[
S_2 = \mathbb P\{ -c_1+c_2 = 0\big| Y = -1, \alpha + X'\beta = -1 \} ,
\]
which equals one if and only if $c_1=c_2$. 
Hence, $S_1= S_2 =1$ if and only if $c_1=c_2=0$, which holds if and only if $\theta_\bullet = 0$. This contradicts against $\theta_\bullet \neq 0$, and proves that at least one of $S_1, S_2$ is strictly less than 1.

For case (II) $\beta_\bullet\notin \text{span}(\beta)$, suppose that $S_1=1$. Then 
\begin{align*}
\mathbb P\big\{ \theta_\bullet' \tilde X \neq 0  \big| Y = 1, \alpha + X'\beta = 1  \big\} = 1-S_1 = 0, 
\end{align*}
which implies 
\begin{align}\label{pf-hessian-pd-aux-1}
& \mathbb P\big\{ \theta_\bullet' \tilde X \neq 0, Y = 1  \big| \alpha + X'\beta = 1  \big\} \notag \\
& = \mathbb P\big\{ \theta_\bullet' \tilde X \neq 0  \big| Y = 1, \alpha + X'\beta = 1  \big\}  \mathbb P \big\{Y = 1\big| \alpha + X'\beta = 1  \big\} =0. 
\end{align}
By Assumption~\ref{aspn-pos-density-Xm} (a), (c) and that $\beta_\bullet \notin \text{span}(\beta)$, 
\[
\mathbb P\big\{ \alpha_\bullet + \beta_\bullet'X \neq 0 \big| \alpha+X'\beta = 1\big\} >0. 
\]
It follows that
\begin{align*}
& \mathbb P\big\{ \theta_\bullet' \tilde X \neq 0, Y = 1  \big| \alpha + X'\beta = 1  \big\} \\
& = \mathbb P\big\{ Y = 1  \big| \theta_\bullet' \tilde X \neq 0, \alpha + X'\beta = 1  \big\} \mathbb P\big\{ \alpha_\bullet + \beta_\bullet'X \neq 0 \big| \alpha+X'\beta = 1\big\} 
>0
\end{align*}
since $\mathbb P\big\{ Y = 1  \big| \theta_\bullet' \tilde X \neq 0, \alpha + X'\beta = 1  \big\}>0$ by Assumption \ref{aspn-pos-density-Xm} (b). 
This contradicts against \eqref{pf-hessian-pd-aux-1}, and thus, at least one of $S_1$ or $S_2$ is strictly less than 1 under Case (II).
The proof is then complete.

\section{Proofs for Section \ref{section-consistency}}

\begin{proof}[Proof of Lemma \ref{lemma-cast-positive}]

For any $c\neq 0$ and any $r\in \mathbb R$, 
\begin{align*}
\dot Q_\bullet(c,r) = - \mathbb{E}\Big(  1\{U\leq V,\,c  V+r < 1\}  - 1\{U>V,\,c V+r > -1\} \Big)\begin{pmatrix} 1 \\ V \end{pmatrix}  
\end{align*}
and $\ddot Q_\bullet(c,r)$ is positive definite by essentially the same arguments as in the proof of Lemma \ref{lemma-gradient-hessian-svm}. 
Thus, $Q_\bullet(c,r)$ has a unique minimum at some $c_\ast>0,r_\ast \in \mathbb R$ over $c,r\in \mathbb R$ if and only if 
\begin{equation*} 
\dot Q_\bullet(c_\ast,r_\ast) =  - \mathbb{E}\Big(  1\{U\leq V,\,c_\ast V+r_\ast < 1\}  - 1\{U>V,\,c_\ast V+r_\ast > -1\} \Big)\begin{pmatrix} 1 \\ V \end{pmatrix}  = 0 
\end{equation*}  
for some $c_\ast>0,r_\ast \in \mathbb R$, 
which holds if and only if 
\begin{align}\label{pf-foc-c-r-soln-aux-1}
\mathbb{E}\left[ 1\{U \leq V,\, V< v_u \} \begin{pmatrix}
1 \\ V \end{pmatrix} \right]
= \mathbb{E}\left[  1\{ U>V,\, V>v_l\} \begin{pmatrix} 1 \\ V \end{pmatrix} \right] 
\end{align}
for some $v_u, v_l \in \mathbb R$ with $v_u > v_l$.

Define
\begin{align*}
 p(v) = \mathbb P\big\{U \leq V<v\big\}
\quad \mbox{and}\quad
 q(v) = \mathbb P\big\{U>V>v\big\}. 
\end{align*} 
By Assumption \ref{aspn-pos-density-Xm} (a)-(b),  
\begin{equation}\label{pf-foc-c-r-soln-aux-2}
\begin{split}
& p \  \text{is strictly increasing},  \quad \lim_{v \to -\infty } p(v) = 0, \  \lim_{v \to \infty } p(v)= \mathbb P\{Y = 1\} \ge 1/2 \\ 
& q \ \text{is strictly decreasing},  \quad \lim_{v \to -\infty } q(v) = \mathbb P\{Y=-1\} \le 1/2, \  \lim_{v \to \infty } q(v) = 0 
\end{split}
\end{equation}
Thus, $q^{-1}\big( p(v)\big)$ is well-defined for any $v<\bar v$. Moreover, 
\begin{equation}
\begin{split}
& \lim_{v\to -\infty} q^{-1}\big( p(v)\big)  = \infty, \quad \lim_{v\to \bar v -} q^{-1}\big( p(v)\big)  = -\infty \\
& \quad  q^{-1}\big( p(v)\big) \text{ is continuous and strictly decreasing in } v\in (-\infty, \bar v) .
\end{split}
\end{equation} 
Thus, the first equation in \eqref{pf-foc-c-r-soln-aux-1} holds if and only if 
\[
v_l = q^{-1}\big(  p(v_u ) \big)  \quad \mbox{and}\quad v_u \in  (-\infty, \bar v).
\]
Define 
\[
h(v) = \mathbb E 1\{U \leq V<v\} V - \mathbb E 1\{U>V>q^{-1}\big( p(v) \big) \} V
\] 
Then \eqref{pf-foc-c-r-soln-aux-1} is satisfied if and only if $h(\cdot)$ has a root in $(-\infty, \bar v)$. 
Thus, it suffices to show that $h(\cdot)$ has a root in $(-\infty, \bar v)$ if and only if Assumption \ref{aspn-suff-cond-consistency-domclass} holds.

Notice 
\begin{align}\label{pf-foc-c-r-soln-aux-h-property-1}
\lim_{v\to -\infty } h(v) = 0. 
\end{align}
Let $v_0$ be such that $p(v_0) = q(v_0)$, whose existence and uniqueness are ensured by \eqref{pf-foc-c-r-soln-aux-2}. By the definition of $v_0$, $q^{-1}\big( p(v_0)\big) = v_0$, i.e., $\mathbb E 1\{U \leq V<v_0\} = \mathbb E 1\{ U>V>v_0\}$. 
Thus, 
\begin{align}\label{pf-foc-c-r-soln-aux-h-property-2}
 &  h(v_0) = \mathbb E 1\{ U \leq V<v_0\} V - \mathbb E 1\{U>V>v_0\} V  \notag \\
& = \mathbb E 1\{ U \leq V<v_0\} V - \mathbb E 1\{U>V>v_0\} V - v_0 \Big( \mathbb E 1\{U \leq V<v_0\} - \mathbb E 1\{ U>V>v_0\} \Big) \notag \\
& = \mathbb E 1\{ U \leq V<v_0\} (V -v_0) + \mathbb E 1\{U>V>v_0\} (v_0-V) \notag \\
& < 0 , 
\end{align}
where the last inequality is because $V$ has positive Lebesgue density on $\mathbb R$ and $\Pi(V)\in (0,1)$ almost everywhere by Assumption~\ref{aspn-pos-density-Xm}(b).  
Now we consider the monotonicity of $h$. 
Let $v_2, v_1$ be fixed arbitrarily with $v_2>v_1$.  Then $q^{-1}(p(v_2)) < q^{-1}(p(v_1))$ by \eqref{pf-foc-c-r-soln-aux-2}, and it follows from 
\begin{align*}
\mathbb E 1\{U \leq V<v_1\} & = \mathbb E 1\{U>V>q^{-1}(p(v_1)) \} \\
\mathbb E 1\{U \leq V<v_2\} & = \mathbb E 1\{U>V>q^{-1}(p(v_2)) \} 
\end{align*}
that 
\begin{align*}
\mathbb E 1\{U \leq V\} 1\{ v_1 \leq V<v_2\} = \mathbb E 1\{U>V\} 1\big\{ q^{-1}(p(v_2)) < V\leq q^{-1}(p(v_1) \big\}. 
\end{align*}
Thus, 
\begin{align}\label{pf-foc-c-r-soln-aux-h-property-3}
& h(v_2) - h(v_1) \notag  \\
& = \mathbb E 1\{U \leq V\} 1\{v_1 \leq V<v_2\} V -  \mathbb E 1\{U>V\} 1\{ q^{-1}\big(p(v_2)\big) <V \leq q^{-1}\big(p(v_1)\big) \} V  \notag \\
& = \mathbb P\{U \leq V, v_1 \leq V<v_2\} \notag  \\ 
& \qquad \times \Big( \mathbb E\big(V|  U \leq V, v_1 \leq V<v_2 \big)  -  \mathbb E\big(V|  U>V, q^{-1}\big(p(v_2)\big) <V \leq q^{-1}\big(p(v_1) \big) \Big)  \notag  \\
& \begin{cases}
    >0 & \text{ if } v_2>v_1>v_0 \\
    <0 & \text{ if } v_1<v_2 <v_0
\end{cases}
\end{align}
where the last inequality is because $q^{-1}(p(v_1)) < v_0 <v_1$ when $v_0<v_1<v_2$, and $v_2 <v_0 < q^{-1}(p(v_2))< q^{-1}(p(v_1))$ when $v_1<v_2 < v_0$. 
\eqref{pf-foc-c-r-soln-aux-h-property-1}, \eqref{pf-foc-c-r-soln-aux-h-property-2}, \eqref{pf-foc-c-r-soln-aux-h-property-3} show that
\begin{align*}
\lim_{v\to -\infty } h(v) & = 0, \quad h(\cdot) \text{ decreasing on } (-\infty, v_0) \\  
h(v_0) & <0 , \quad  h(\cdot) \text{ increasing on } (v_0, \bar v)  
\end{align*}
Therefore, $h(\cdot)$ has a root in $(-\infty, \bar v)$ if and only if $h(\bar v) >0$, which was to be shown. 
\end{proof}

\begin{proof}[Proof of Lemma \ref{lemma-v-bar}] 
Note that Assumption~\ref{aspn-med0-error} implies that $\Pi(v) \geq 0.5$ if $v\geq 0$ and $\Pi(v)<0.5$ if $v<0$, which also implies $\Pi(v) < 1-\Pi(v)$ for all $v<0$. 
Moreover, $\mathcal L(V)$ has positive Lebesgue density on $\mathbb R$ by Assumption \ref{aspn-pos-density-Xm}(a). 
Recall in Assumption~\ref{aspn-suff-cond-consistency-domclass}, $\overline v = (p^{-1}\circ q)(-\infty)$ is defined as $\mathbb P \{U<V<\overline v\} = \mathbb P\{U>V\}$, which reduces to $h(\overline v)=0$, where
\begin{align}\label{pf-v-bar-aux}
h(v) =  \mathbb E\,\Pi(V)1\{V<v\}  - \mathbb E\big[1-\Pi(V)\big]
\end{align} 
by Assumption~\ref{id-dep}. Note that $h(v)$ is strictly increasing in $v$, since $\mathcal L(V)$ has positive Lebesgue density on $\mathbb R$ and $\Pi(V) \in (0,1)$ almost everywhere by Assumption~\ref{aspn-pos-density-Xm}(b). 

For Part (a), $\Pi(v) < 1-\Pi(v)$ for all $v<0$ and $\mathcal L(V)$ having positive Lebesgue density on $\mathbb R$ imply that $\mathbb E\,\Pi(V)1\{V<0\} < \mathbb E\big[1-\Pi(V)\big]1\{V<0\}$, where $\mathbb E\big[1-\Pi(V)\big]1\{V<0\} \leq \mathbb E\big[1-\Pi(V)\big]$ because $1-\Pi(v) \geq 0$ for $v>0$. 
Thus, $h(0)<0$, which implies that $\overline v>0$, since $h(v)$ is strictly increasing. 

For Part (b), denote $\mathcal L(V)$ by $P$. As $P$ decreases for $\Pi$ down to 1/2, i.e., $\int \Pi(v) P(dv)$ decreases to 1/2, 
\[
 \int \Pi(v)P(dv) - \int \big[1-\Pi(v)\big]P(dv)
\]
decreases down to 0, and therefore, $\overline v(P)$ defined as
\[
 \int \Pi(v)1\big\{v<\overline v(P)\big\}P(dv) - \int \big[1-\Pi(v)\big]P(dv) = 0
\]
increases. If $P$ decreases for $\Pi$ down to the level 
\[
 \int \Pi(v)P(dv) - \int \big[1-\Pi(v)\big]P(dv) = 0,
\]
then $\overline v(P) = \infty$, since $\mathcal L(V)$ has positive Lebesgue density on $\mathbb R$.  
This was to be shown.

Finally, for Part (c), note that
\[
 v\Pi(v) \geq v\big[1-\Pi(v)\big]
\]
for all $v\in\mathbb R$, with strict inequality for all $v<0$. It follows immediately that $\tau(\infty) = \mathbb E\,V\Pi(V) > \mathbb E\,V\big[1-\Pi(V)\big] = \sigma(-\infty)$, since $\mathcal L(V)$ has positive Lebesgue density. The proof is therefore complete.
\end{proof}

\begin{proof}[Proof of Proposition \ref{proposition-alpha-ms-consistency}]
Denote by $\beta_\ast$ the probability limit of the SVM slope estimator $\hat \beta$. Under the conditions in Theorem \ref{theorem-consistency-svm}, $\beta_\ast = c_\ast \beta_0$ 
for some $c_\ast >0$. 

We first show that
\begin{align}\label{pf-alpha-ms-obj-unif-cgce}
\sup_{\alpha\in \mathcal A } \big| \hat Q_{ms}(\alpha,\hat \beta) - \hat Q_{ms}(\alpha,\beta_\ast) \big| \to_p 0.
\end{align}
Notice that 
\begin{align*}
\big| & \hat Q_{ms}(\alpha,\hat \beta) - \hat Q_{ms}(\alpha,\beta_\ast) \big| \leq  \frac{1}{n} \sum_{i=1}^n 1\left\{ (\alpha+ \hat\beta'X_i)(\alpha+\beta_\ast'X_i) \leq 0   \right\} \\
& =  \frac{1}{n} \sum_{i=1}^n 1\left\{ (\alpha+ \beta_\ast'X_i)^2 \leq - (\alpha+\beta_\ast'X_i) X_i(\hat\beta - \beta_\ast)  \right\} \\ 
& \leq \frac{1}{n} \sum_{i=1}^n 1\left\{ |\alpha+ \beta_\ast'X_i|\leq \|X_i\| \|\hat\beta - \beta_\ast\| \right\}
\end{align*}
Then for any $\epsilon, \eta>0$, 
\begin{align}\label{pf-alpha-ms-aux-1}
\mathbb P &  \left\{ \sup_{\alpha\in \mathcal A } \big| \hat Q_{ms}(\alpha,\hat \beta) - \hat Q_{ms}(\alpha,\beta_\ast) \big| > \epsilon \right\} \notag \\
& \leq \mathbb P  \left\{ \|\hat\beta - \beta_\ast\|>\eta \right\} + \mathbb P  \left\{ \sup_{\alpha\in \mathcal A } \frac{1}{n} \sum_{i=1}^n 1\left\{ |\alpha+ \beta_\ast'X_i| \leq \|X_i\| \eta  \right\} > \epsilon \right\}. 
\end{align}
where the first term is arbitrarily small by taking $n$ large enough, since $\hat\beta\to_p \beta_\ast$. For the second term, notice 
we can choose $\eta$ small such that
\[
\sup_{\alpha\in \mathcal A } \mathbb E 1\left\{ |\alpha+ \beta_\ast'X_i|  \leq \|X_i\| \eta  \right\}  < \epsilon/2, 
\]
since $\mathcal A$ is compact and $\beta_\ast'X_i$ has a Lebesgue density. 
Then
\begin{align}\label{pf-alpha-ms-aux-2}
\mathbb P & \left\{ \sup_{\alpha\in \mathcal A } \frac{1}{n} \sum_{i=1}^n 1\left\{ |\alpha+ \beta_\ast'X_i| \leq \|X_i\| \eta  \right\} > \epsilon \right\}  \notag \\
& \leq \mathbb P  \left\{ \sup_{\alpha\in \mathcal A } \left( \frac{1}{n} \sum_{i=1}^n 1\left\{ |\alpha+ \beta_\ast'X_i|  \leq \|X_i\| \eta  \right\} - \mathbb E 1\left\{ |\alpha+ \beta_\ast'X_i| \leq\|X_i\| \eta  \right\} \right)  > \epsilon/2 \right\} \notag \\
& \to 0 \quad \mbox{as} \quad n\to \infty,
\end{align} 
where the last line follows from the Glivenko-Cantelli theorem (see, e.g., Theorem~2.4.3 and 2.6.4 in \citealp{vandervaart-wellner-96}), once provided with
\begin{align}\label{pf-alpha-ms-vc-class}
   \{x \mapsto 1\left\{ \|\alpha+\beta_\ast'x| \leq \|x\|\eta \right\} | \alpha\in \mathcal A\} \quad  \mbox{is a VC-class} .
\end{align}
Consequently, the second term in \eqref{pf-alpha-ms-aux-1} can also be made arbitrarily small by \eqref{pf-alpha-ms-aux-2}, and thus, \eqref{pf-alpha-ms-obj-unif-cgce} holds as was to be shown.  
It remains to verify \eqref{pf-alpha-ms-vc-class} to complete the proof for \eqref{pf-alpha-ms-obj-unif-cgce}. Note that $ \{ x\mapsto \alpha+\beta_\ast'x | \alpha\in \mathbb R\} $ and $\{ x\mapsto -(\alpha+\beta_\ast'x) | \alpha\in \mathbb R\} $ are VC-classes by Lemma 2.6.15 in \cite{vandervaart-wellner-96}, and thus, $ \{ x\mapsto |\alpha+\beta_\ast'x| = \max\{ \alpha+\beta_\ast'x, -(\alpha+\beta_\ast'x)\} | \alpha\in \mathbb R\} $ is VC by Lemma~2.6.18 (ii) in \cite{vandervaart-wellner-96}. Consequently, $ \{ x\mapsto |\alpha+\beta_\ast'x| - \|x\|\eta  | \alpha\in \mathbb R\} $ is VC, since $x\mapsto \|x\| \eta$ is a fixed function and by Lemma~2.6.18 (v) in \cite{vandervaart-wellner-96}, and thus, $\{x \mapsto 1\left\{ \|\alpha+\beta_\ast'x| \leq \|x\|\eta \right\} | \alpha\in  \mathbb R\}$ is VC by Lemma~2.6.18 (iii) and (iv) in \cite{vandervaart-wellner-96}.

Now we show $\hat \alpha_{ms} \to_p c_\ast \alpha_0$. Note that
\[
\sup_{\alpha\in \mathcal A }  \big| \hat Q_{ms}(\alpha,\beta_\ast) - Q_{ms}(\alpha,\beta_\ast) \big| \to_p 0 
\]
by Lemma 4 in \cite{manski-85}. Together with \eqref{pf-alpha-ms-obj-unif-cgce}, we have
\begin{align*} 
& \sup_{\alpha\in \mathcal A } \big| \hat Q_{ms}(\alpha,\hat \beta) - Q_{ms}(\alpha,\beta_\ast) \big| \\
& \leq \sup_{\alpha\in \mathcal A } \big| \hat Q_{ms}(\alpha,\hat \beta) - \hat Q_{ms}(\alpha,\beta_\ast) \big| + \sup_{\alpha\in \mathcal A }  \big| \hat Q_{ms}(\alpha,\beta_\ast) - Q_{ms}(\alpha,\beta_\ast) \big|   \to_p 0.
\end{align*}
Moreover, by Assumptions \ref{aspn-med0-error}, \ref{aspn-pos-density-Xm} and Lemma 3 in \cite{manski-85}, $\alpha_\ast := c_\ast \alpha_0$ is the unique minimum of $Q_{ms}(\alpha,\beta_\ast)$. 
Since $\mathcal A$ is compact and $Q_{ms}(\alpha,\beta_\ast)$ is continuous in $\alpha$ (see, e.g., Lemma 5 in \citealp{manski-85}), $\hat \alpha_{ms} \to_p c_\ast \alpha_0$ by Theorem 2.1 in \cite{newey-mcfadden-94}.   
\end{proof}

\begin{proof}[Proof of Theorem \ref{theorem-alpha-ms-equiv-limit-distribution}]

\cite{chen-lee-sung-14} study maximum score estimation when some regressors are conditional expectations, which are not observable and need to be estimated nonparametrically in the first stage. They show that the asymptotic distribution of the second stage MS estimator is the same as if the conditional expectations were used, provided that the conditional expectations are estimated with a convergence rate faster than $n^{1/3}$ in the first stage, along with some other technical conditions.

We will apply the results in \cite{chen-lee-sung-14} to prove Theorem \ref{theorem-alpha-ms-equiv-limit-distribution}. 
For this purpose, we need some re-parameterizations to put our second-step intercept estimator in their framework. 
Let $X_i = (X_{-m,i}', X_{m,i})'$ and $\beta = (\beta_{-m}',\beta_{m})'$ to separate the first $m-1$ components and the $m$-th component. Recall that we assumed $\beta_{m,0}>0$ without loss of generality. 
Define $\tilde{\mathcal A} = \{ \alpha/\beta_{m,\ast} | \alpha \in \mathcal A\}$
\[
\tilde \alpha_{\ast} = \arg \max_{\alpha\in \tilde{\mathcal A}} \frac{1}{n} \sum_{i=1}^n Y_i 1\bigg\{ \alpha+X_{-m,i}'\frac{\beta_{-m,0}}{\beta_{m,0}} + X_{m,i} \geq 0 \bigg\}  
\]
and note that $\tilde \alpha_{\ast} = \frac{\hat \alpha_{\ast}}{\beta_{m,\ast}}$. 
We further define 
\[
\tilde \alpha_{ms} = \arg \max_{\alpha\in \tilde{\mathcal A}} \frac{1}{n} \sum_{i=1}^n Y_i 1\bigg\{ \alpha+ X_{-m,i}'\frac{\hat \beta_{-m}}{\hat\beta_{m}} + X_{m,i}\geq 0 \bigg\}  
\]
and note that $\tilde \alpha_{ms} = \frac{\hat \alpha_{ms}}{\hat\beta_{m}}$ with probability approaching one, since $\hat\beta_{m} \to_p c_\ast \beta_{m,0} >0$. 
Thus, we can view $\hat G(x_{-m}) :=  x_{-m}'\frac{\hat \beta_{-m}}{\hat\beta_{m}}$ as an estimator for $G(x_{-m}):= x_{-m}' \frac{\beta_{-m,0}}{\beta_{m,0}}$ in the framework of \cite{chen-lee-sung-14}. 
Let $\mathcal X_{-m}$ denote the support of $X_{-m}$, which is assumed to be bounded in Assumption \ref{aspn-alpha-ms-rate}.

Applying Theorem 4.3 in \cite{chen-lee-sung-14} shows that 
\begin{equation}\label{pf-alpha-limit-distn-1}
n^{1/3}\left(\tilde \alpha_{ms} - \frac{\alpha_0}{\beta_{m,0}} \right) = n^{1/3}\left(\tilde \alpha_{\ast} - \frac{\alpha_0}{\beta_{m,0}} \right) + o_p(1)
\end{equation}
provided with Assumptions \ref{aspn-pos-density-Xm}, \ref{aspn-alpha-ms-rate}, and the following conditions which will be verified at the end of the proof. 
\begin{enumerate}[(i)] 
\itemsep -0.1mm
\item $G(x_{-m})$ has derivatives at least up to order $m-1$, and all of the derivatives are uniformly bounded by a constant; 
\item $\|\hat G-G\|_\infty := \sup_{x_{-m}\in {\mathcal X}_{-m}} |\hat G(x_{-m})-G(x_{-m})| =o_p(n^{-1/3})$. In addition, all derivatives of $\hat G$, up to order $m-1$, converge to the corresponding derivatives of $G$ under the uniform norm with rate faster than $n^{-1/3}$;
\item The distribution function $F_{\alpha + X'\beta_0/\beta_{m,0}}(\cdot)$ of $\alpha+X_{-m}'\frac{\beta_{-m,0}}{\beta_{m,0}} + X_{m} = \alpha + X'\frac{\beta_0}{\beta_{m,0}}$ is uniformly Lipschitz over $\alpha \in \tilde{\mathcal A}$. 
That is, there exists a constant $L>0$, which does not depend on $\alpha$, such that $|F_{\alpha + X'\beta_0/\beta_{m,0}}(u) - F_{\alpha + X'\beta_0/\beta_{m,0}}(v)| \leq L|u-v|$ for any $u,v\in \mathbb R$ and $\alpha\in \tilde{\mathcal A}$. 

\end{enumerate} 
Given \eqref{pf-alpha-limit-distn-1}, multiplying by $c_\ast \beta_{m,0} = \beta_{m,\ast}$ on both sides yields
\begin{equation*}
n^{1/3}\left( \hat \alpha_{ms} \frac{\beta_{m,\ast}}{\hat\beta_{m} }  - \alpha_\ast \right) = n^{1/3}\left( \hat\alpha_{\ast} - \alpha_\ast \right) + o_p(1),
\end{equation*}
and thus, 
\begin{align*}
n^{1/3}\left( \hat\alpha_{\ast} - \alpha_\ast \right) & = n^{1/3}\left( \hat \alpha_{ms} \frac{\beta_{m,\ast}}{\hat\beta_{m} }  - \alpha_\ast \right) + o_p(1) \\
& = n^{1/3}\left( \hat \alpha_{ms} - \alpha_\ast \right) +  \hat \alpha_{ms} n^{1/3}\left( \frac{\beta_{m,\ast}}{\hat\beta_{m} }  - 1  \right) + o_p(1) \\
& = n^{1/3}\left( \hat \alpha_{ms} - \alpha_\ast \right) + o_p(1) 
\end{align*}
where the last line follows from $\hat\beta \to_p \beta_{m,\ast}$, $\beta_{m,\ast}\neq 0$, and that  $\hat \alpha_{ms} \in \mathcal A$ a compact set in $\mathbb R$. This proves Theorem \ref{theorem-alpha-ms-equiv-limit-distribution}. 

It remains to verify conditions (i)-(iii).  
Condition (i) holds trivially due to the form $G(x_{-m}) = x_{-m}' \frac{\beta_{-m,0}}{\beta_{m,0}}$. 
Notice $\|\hat \beta - \beta_\ast\| = O_p(n^{-1/2})$ and $\beta_{m,\ast} = c_\ast \beta_{m,0} \neq 0$ imply
\begin{align*} 
\left\|\frac{\hat \beta}{\hat\beta_{m}} - \frac{\beta_{0}}{ \beta_{m,0} } \right\| = O_p(n^{-1/2}). 
\end{align*}
and thus, Condition (ii) is satisfied upon noticing that
\begin{align*} 
\sup_{ x_{-m}\in \mathcal X_{-m} } |\hat G(x) - G(x)| & \leq  \sup_{x_{-m}\in \mathcal X_{-m}} \|x_{-m}\|  \left\|\frac{\hat \beta_{-m}}{\hat\beta_{m}} - \frac{\beta_{-m,0}}{\beta_{m,0} } \right\|     \notag \\
& = O_p(n^{-1/2})
\end{align*}
since $\mathcal X_{-m}$ is assumed to be bounded. 
Moreover, Condition (iii) is satisfied trivially, since $\mathcal A$ is compact and the Lebesgue density $p(x_m|x_{-m})$ of $\mathcal{L}(X_m|X_{-m}=x_{-m})$ satisfies $p(x_m|x_{-m})< K$ for all $(x_m,x_{-m}) \in \mathcal X$. 
This completes the proof. 
\end{proof}

\section{Proofs for Section \ref{section-limit-distn}}

We introduce some notations for the convenience of exposition. 
Denote by $P$ the common distribution of $(Y_i,X_i')'$, and $P_n$ the empirical distribution. We write $Pf$ for the expectation $\mathbb{E}f(Y,X) = \int f dP$, and abbreviate the average $n^{-1}\sum_{i=1}^n f(Y_i,X_i)$ to $P_n f$. We also abbreviate the centered sums $n^{-1/2} \sum_{i=1}^n \big( f(Y_i,X_i) - \mathbb{E}f(Y_i,X_i) \big)$ to $G_n f$. 
Let 
\[
\ell_\theta(y,x) = \Big( 1- y(\alpha+x'\beta) \Big)_+ ,
\]
and define
\[
\dot \ell_{\theta}(y,x) = -1\{ 1-y(\alpha + x'\beta) >0 \} y \begin{pmatrix}
1 \\ x
\end{pmatrix}.
\] 
Note that $\mathbb{E} \dot \ell_{\theta_\ast}(Y,X) = \dot{Q}(\theta_\ast) = 0$ since $\beta_{m,\ast}\neq 0$ and Lemma~\ref{lemma-gradient-hessian-svm} imply that $Q$ is differentiable at $\theta_\ast$ with gradient $\dot Q(\theta_\ast)$.  
Recall 
\begin{align*}
J  
& =  \mathbb{E} 1\{1-  Y (\alpha_\ast + X'\beta_\ast) > 0\} \begin{pmatrix}
1 & X' \\ X & XX'
\end{pmatrix},
\end{align*}
where $J < \infty$ since $\mathbb{E} \|X\|^2 <\infty$.  
Thus,
\begin{align*}
G_n \dot \ell_{\theta_\ast} & = \frac{1}{\sqrt{n}}\sum_{i=1}^n \Big( \dot \ell_{\theta_\ast} (Y_i,X_i) - \mathbb{E} \dot \ell_{\theta_\ast} (Y_i,X_i)   \Big) \\
& = \frac{1}{\sqrt{n}}\sum_{i=1}^n  \dot \ell_{\theta_\ast} (Y_i,X_i)  \to_d N(0, J ) 
\end{align*}
as $n \to \infty$. 
Note that $Q(\theta) = P\ell_{\theta}$, $Q_n(\theta) = P_n \ell_{\theta} + \frac{\lambda_n}{n} \|\beta\|^2$, and 
\[
\sqrt{n}\Big( (Q_n - Q )(\theta) - (Q_n - Q )(\theta_\ast) \Big) =  G_n(\ell_{\theta} -\ell_{\theta_\ast}) + \frac{\lambda_n}{\sqrt{n}} \Big( \|\beta\|^2 - \|\beta_{\ast}\|^2 \Big).
\]

\begin{proof}[Proof of Lemma \ref{lemma-svm-cgce-rate-aux}]
By Lemma \ref{lemma-gradient-hessian-svm} and $\beta_{m,\ast}\neq 0$, $P \ell_{\theta}$ admits a second-order Taylor expansion at the point of minimum $\theta_\ast$ a with nonsingular Hessian matrix $\ddot{Q}(\theta_\ast)$. 
For any $\theta_1 = (\alpha_1,\beta_1')'$ and $\theta_2= (\alpha_2,\beta_2')'$, we have 
\begin{align}\label{m_theta-lipschitz}
|\ell_{\theta_1}(y,x) -\ell_{\theta_2}(y,x) | & = \Big| \Big( 1- y(\alpha_1+x'\beta_1) \Big)_+ - \Big( 1- y(\alpha_2+x'\beta_2) \Big)_+ \Big| \notag \\
&\leq \Big| y(\alpha_1+x'\beta_1) - y(\alpha_2+x'\beta_2) \Big| = \Big|\alpha_1 - \alpha_2 + x'(\beta_1 -\beta_2)\Big|  \notag  \\
&\leq \sqrt{ 1+ \|x\|^2} \|\theta_1-\theta_2\|,  
\end{align}
where the first inequality is due to the fact that $|a_+-b_+| \leq |a-b|$ for any real numbers $a,b$. Thus, $\theta \mapsto \ell_\theta$ is Lipschitz.
Together with $\mathbb{E}(1+\|X\|^2) <\infty$, it implies that,  
for every $n$ and every sufficiently small $\delta>0$, 
\begin{align}\label{eq-modulus-aux-1}
\mathbb{E} \sup_{ \|\theta-\theta_{\ast}\| <\delta } \sqrt{n} \big| (P_n - P )(\ell_\theta) - (P_n - P )(\ell_{\theta_{\ast}}) \big|  \lesssim \delta  
\end{align}
by the proof in Corollary 5.53 of \cite{vandervaart-00}, where the notation $\lesssim$ means ``is bounded above up to a universal constant''.

Moreover, since
\begin{align*}
\Big| \|\beta\|^2 -\|\beta_{\ast}\|^2 \Big| & =  \Big( \|\beta\|  +\|\beta_{\ast}\| \Big) \big| \|\beta\|  -\|\beta_{\ast}\| \big|   \leq \Big( 2 \|\beta_{\ast}\| + \|\beta - \beta_{\ast}\|  \Big) \|\beta - \beta_{\ast}\| \\
& \leq \Big( 2 \|\beta_{\ast}\| + \|\theta - \theta_{\ast}\|  \Big) \|\theta - \theta_{\ast}\|,
\end{align*}
we have
\begin{align}\label{eq-modulus-aux-2}
\frac{\lambda_n}{\sqrt{n}} \sup_{ \|\theta-\theta_{\ast}\| <\delta } \Big| \|\beta\|^2 -\|\beta_{\ast}\|^2 \Big|
\leq \frac{\lambda_n}{\sqrt{n}}   \big( 2 \|\beta_{\ast}\| + \delta \big) \delta \lesssim \delta,
\end{align}
for every $n$ and every $\delta>0$ sufficiently small, since $\lambda_n/\sqrt{n} =o(1)$ and $\beta_{\ast} \neq 0$. 
Combining \eqref{eq-modulus-aux-1} and \eqref{eq-modulus-aux-2} yields
\begin{align*}
\mathbb{E}  & \sup_{ \|\theta-\theta_{\ast}\| <\delta } \sqrt{n} \big| (Q_n - Q )(\theta) - (Q_n - Q )(\theta_{\ast}) \big|  \\
&\leq \mathbb{E} \sup_{ \|\theta-\theta_{\ast}\| <\delta } \sqrt{n} \big| (P_n - P )(\ell_\theta) - (P_n - P )(\ell_{\theta_{\ast}}) \big| + \frac{\lambda_n}{\sqrt{n}} \sup_{ \|\theta-\theta_{\ast}\| <\delta } \Big| \|\beta\|^2 -\|\beta_{\ast}\|^2 \Big| \\
& \lesssim \delta.
\end{align*}
for every $n$ and every sufficiently small $\delta>0$. Lemma \ref{lemma-svm-cgce-rate-aux} follows immediately.
\end{proof}

\begin{proof}[Proof of Lemma \ref{lemma-svm-cgce-rate}]

By Lemma \ref{lemma-gradient-hessian-svm} and that $\beta_{m,\ast}\neq 0$, $Q(\theta)$ is twice continuously differentiable at the point of minimum $\theta_{\ast}$ with positive definite Hessian matrix $H = \ddot{Q}(\theta_{\ast})$. 
Thus, for every $\theta$ in a neighborhood of $\theta_{\ast}$, we have
\[
Q(\theta) -Q(\theta_{\ast}) \gtrsim \|\theta-\theta_{\ast}\|^2.
\]
By Lemma \ref{lemma-svm-cgce-rate-aux},
\begin{align*}
\mathbb{E} \sup_{ \|\theta-\theta_{\ast}\| <\delta } \sqrt{n} \big| (Q_n - Q )(\theta) - (Q_n - Q )(\theta_\ast) \big|  \lesssim \delta.
\end{align*} 
Since $\hat\theta \to_p \theta_\ast$ as established in Lemma \ref{lemma-plim-theta},  
applying Theorem 3.2.5 in \cite{vandervaart-wellner-96} gives immediately that $\sqrt{n}(\hat\theta -\theta_\ast) = O_p(1)$. 
\end{proof}

\begin{proof}[Proof of Theorem \ref{theorem-svm-limit-distn}]

For notational simplicity, we denote by $H := \ddot Q(\theta_\ast)$.

We first show that
\begin{align}\label{svm-limit-distn-aux}
n\Big( Q_n(\theta_\ast + \tilde h_n /\sqrt{n} ) - Q_n(\theta_\ast) \Big) = \frac{1}{2}  \tilde h_n' H  \tilde h_n +  \tilde h_n' G_n \dot \ell_{\theta_\ast} +o_p(1),
\end{align}
for any random sequence $\tilde h_n = O_p(1)$. 

Let $\tilde h_n =O_p(1)$ be an arbitrary sequence of $(m+1)$-dimensional random vector. Denote by $\tilde h_{n,-1}$ the random vector of the first $m$-elements of $\tilde h_{n}$. Notice
\begin{align}\label{limit-distn-aux-decomposition}
n \Big( Q_n(\theta_\ast + \tilde h_n /\sqrt{n} ) - Q_n(\theta_\ast) \Big) & = A_n +B_n +C_n + \frac{1}{2}\tilde h_n' H \tilde h_n + \tilde h_n' G_n \dot \ell_{\theta_\ast} 
\end{align}
where
\begin{align*}
A_n & := G_n\Big( \sqrt{n}\big(  \ell_{\theta_\ast+\tilde h_n/\sqrt{n}} - \ell_{\theta_\ast} \big) \Big) -\tilde h_n' G_n \dot \ell_{\theta_\ast} \\
B_n & := nP\big( \ell_{\theta_\ast+\tilde h_n/\sqrt{n} } -\ell_{\theta_\ast} \big)  -\frac{1}{2}\tilde h_n' H \tilde h_n \\
C_n &:= \lambda_n\Big( \|\beta_\ast + \tilde h_{n,-1}/\sqrt{n}\|^2 - \|\beta_\ast \|^2\Big).
\end{align*}  
Since $\theta \mapsto \ell_{\theta}$ is Lipschitz in the sense of \eqref{m_theta-lipschitz}, applying Lemma 19.31 in \cite{vandervaart-00} gives that 
\begin{align}\label{limit-distn-aux-An}
A_n  = G_n\Big( \sqrt{n}\big(  \ell_{\theta_\ast+\tilde h_n/\sqrt{n}} - \ell_{\theta_\ast} \big) \Big) -\tilde h_n' G_n \dot \ell_{\theta_\ast} 
& = o_p (1)
\end{align}
for any random sequence $\tilde h_n =O_p(1)$. 
For any $K>0$,
\begin{align}\label{limit-distn-aux-Cn-aux}
\sup_{\|g\| <K } & \Big|  \lambda_n\Big( \|\beta_\ast + g/\sqrt{n}\|^2 - \|\beta_\ast \|^2\Big) \Big|  \leq \sup_{\|g\| <k } \lambda_n \Big( 2\|\beta_\ast\|^2 + \frac{\|g\|^2}{\sqrt{n}} \Big) \frac{\|g\|}{\sqrt{n}} \notag \\
& \leq \frac{\lambda_n}{\sqrt{n}} \Big( 2\|\beta_\ast\|^2 + \frac{K^2}{\sqrt{n}} \Big) K = o(1),
\end{align}
due to $\lambda_n = o(\sqrt{n})$, and
\begin{align}\label{limit-distn-aux-Bn-aux}
\sup_{\|h\|<K} \Big| nP\big( \ell_{\theta_\ast+h/\sqrt{n} } -\ell_{\theta} \big)  -\frac{1}{2}h' H h \Big|
= \sup_{\|h\|<K} o( \|h\| ) = o(1),
\end{align}
which follows from that $\theta \mapsto P\ell_\theta =Q(\theta)$ is twice continuously differentiable at $\theta_\ast$ with gradient $\dot{Q}(\theta_\ast) =0$ and positive definite Hessian matrix $H = \ddot{Q}(\theta_\ast)$, as established in Lemma \ref{lemma-gradient-hessian-svm}. 
Therefore, for any $\tilde h_n =O_p(1)$, it follows from  \eqref{limit-distn-aux-Bn-aux} and the stochastic boundedness $\tilde h_n =O_p(1)$ that
\begin{align}\label{limit-distn-aux-Bn}
B_n = nP\big( \ell_{\theta_\ast+\tilde h_n/\sqrt{n} } -\ell_{\theta_\ast} \big)  -\frac{1}{2}\tilde h_n' H \tilde h_n  
& =o_p(1)
\end{align}
and from \eqref{limit-distn-aux-Cn-aux} that
\begin{align}\label{limit-distn-aux-Cn}
C_n  = \lambda_n\Big( \|\beta_\ast + \tilde h_{n,-1}/\sqrt{n}\|^2 - \|\beta_\ast \|^2\Big)  
& =o_p(1).
\end{align}
Combining \eqref{limit-distn-aux-An}, \eqref{limit-distn-aux-Bn}, \eqref{limit-distn-aux-Cn} with \eqref{limit-distn-aux-decomposition} yield
\begin{align*}
n \Big( Q_n(\theta_\ast + \tilde h_n /\sqrt{n} ) - Q_n(\theta_\ast) \Big) 
& = \frac{1}{2}\tilde h_n' H \tilde h_n + \tilde h_n' G_n \dot \ell_{\theta_\ast} +o_p(1).
\end{align*}
This proves \eqref{svm-limit-distn-aux}.

Now we prove Theorem \ref{theorem-svm-limit-distn}. For this, we let $\hat h_n = \sqrt{n}(\hat\theta -\theta_\ast)$. Notice $\hat h_n =O_p(1)$ by Lemma \ref{lemma-svm-cgce-rate}, and thus, 
\begin{align}\label{limit-distn-aux-1}
n\Big( Q_n(\theta_\ast + \hat h_n /\sqrt{n} ) - Q_n(\theta_\ast) \Big) = \frac{1}{2}  \hat h_n' H \hat h_n +  \hat h_n' G_n \dot \ell_{\theta_\ast} +o_p(1)
\end{align}
by \eqref{svm-limit-distn-aux}. 
Note also that $H:=\ddot{Q}(\theta_\ast)$ is positive definite by Lemma \ref{lemma-gradient-hessian-svm}, and $G_n \dot \ell_{\theta_\ast} \to_d N(0,J)$ where $J <\infty$. 
Thus, $-H^{-1}G_n \dot \ell_{\theta_\ast} = O_p(1)$. 
Applying \eqref{svm-limit-distn-aux} to $\tilde h_n = -H^{-1}G_n \dot \ell_{\theta_\ast}$ and rearranging terms yield
\begin{align}\label{limit-distn-aux-2}
n\Big( Q_n\big(\theta_\ast -H^{-1}G_n \dot \ell_{\theta_\ast} /\sqrt{n} \big) - Q_n(\theta_\ast) \Big) = - \frac{1}{2} \big(G_n \dot \ell_{\theta_\ast}\big)' H^{-1} G_n \dot \ell_{\theta_\ast} +o_p(1).
\end{align}
Since $\hat h_n$ minimize $Q_n(\theta_\ast + h/\sqrt{n})$, we have
\begin{align*}
n\Big( Q_n(\theta_\ast + \hat h_n /\sqrt{n} ) - Q_n(\theta_\ast) \Big) \leq n\Big( Q_n(\theta_\ast -H^{-1}G_n \dot \ell_{\theta_\ast} /\sqrt{n} ) - Q(\theta_\ast) \Big).
\end{align*}
Together with \eqref{limit-distn-aux-1} and \eqref{limit-distn-aux-2}, we have
\begin{align*}
\frac{1}{2}  \hat h_n' H \hat h_n +  \hat h_n' G_n \dot \ell_{\theta_\ast} +o_p(1) \leq - \frac{1}{2} \big(G_n \dot \ell_{\theta_\ast}\big)' H^{-1} G_n \dot \ell_{\theta_\ast} +o_p(1),
\end{align*}
from which it follows that
\begin{align*}
\frac{1}{2}  \Big( \hat h_n + H^{-1} G_n \dot \ell_{\theta_\ast} \Big)' H \Big( \hat h_n + H^{-1} G_n \dot \ell_{\theta_\ast} \Big) \leq o_p(1).
\end{align*}
Since $H := \ddot{Q}(\theta_\ast)$ is positive definite as established in Lemma \ref{lemma-gradient-hessian-svm}, the quadratic form is non-negative. Hence,
\[
\hat h_n + H^{-1} G_n \dot \ell_{\theta_\ast} \to_p 0.
\]
Therefore,
\begin{align*}
\sqrt{n}(\hat\theta -\theta_\ast)  = \hat h_n  = -H^{-1} G_n \dot \ell_{\theta_\ast}+o_p(1) \to_d N(0, H^{-1}J H^{-1})
\end{align*}
as was to be shown. 
\end{proof}

\section{Proof for Theorem \ref{theorem-consistency-class-wgt-svm}}

\begin{proof}[Proof for Theorem \ref{theorem-consistency-class-wgt-svm}] 
Using the same arguments when proving the consistency of SVM, it suffices to show that, there exists $c>0, r\in \mathbb R$ such that
\begin{align*}
\mathbb E 1\{U\leq V, cV+r <1\} \begin{pmatrix}
    1\\ V
\end{pmatrix} 
= w \mathbb E 1\{U>V, cV+r >-1\} \begin{pmatrix}
    1\\ V
\end{pmatrix} 
\end{align*}
which holds if and only if 
\begin{align}\label{pf-consistency-wgt-svm-aux-1}
\mathbb{E}\left[ 1\{U \leq V,\, V< v_u \} \begin{pmatrix}
1 \\ V \end{pmatrix} \right]
= w\mathbb{E}\left[  1\{ U>V,\, V>v_l\} \begin{pmatrix} 1 \\ V \end{pmatrix} \right] 
\end{align}
for some $v_u, v_l \in \mathbb R$ with $v_u > v_l$.

Following the proof of Lemma \ref{lemma-cast-positive}, define 
\begin{align*}
 p(v) = \mathbb P\big\{U \leq V<v\big\}
\quad \mbox{and}\quad
 q_w(v) = w \mathbb P\big\{U>V>v\big\}. 
\end{align*} 
Note that $\lim_{v \to -\infty } q_w(v) =w\mathbb P\{Y=-1\} =\mathbb P\{Y = 1\}$ by the definition of $w$, and 
\begin{equation}\label{pf-consistency-wgt-svm-aux-2}
\begin{split}
& p \  \text{is strictly increasing},  \quad \lim_{v \to -\infty } p(v) = 0, \  \lim_{v \to \infty } p(v)= \mathbb P\{Y = 1\} \ge 1/2 \\ 
& q_w \ \text{is strictly decreasing},  \quad \lim_{v \to -\infty } q_w(v) =\mathbb P\{Y = 1\}, \  \lim_{v \to \infty } q_w(v) = 0  .
\end{split}
\end{equation}
Thus, $q_w^{-1}\big( p(v)\big)$ is well-defined for any $v \in \mathbb R$, and 
\begin{equation}\label{pf-consistency-wgt-svm-aux-3}
\begin{split}
& \lim_{v\to -\infty} q_w^{-1}\big( p(v)\big)  = \infty, \quad \lim_{v\to \infty} q_w^{-1}\big( p(v)\big)  = -\infty \\
& \quad  q_w^{-1}\big( p(v)\big) \text{ is continuous and strictly decreasing in } v\in \mathbb R .
\end{split}
\end{equation} 
Therefore, the first equation in \eqref{pf-consistency-wgt-svm-aux-1} holds if and only if 
\[
v_l = q_w^{-1}\big(  p(v_u ) \big)  \quad \mbox{and}\quad v_u \in \mathbb R. 
\]
Define 
\[
h_w(v) = \mathbb E 1\{U \leq V<v\} V - w \mathbb E 1\{U>V>q_w^{-1}\big( p(v) \big) \} V .
\] 
Then \eqref{pf-consistency-wgt-svm-aux-1} is satisfied if and only if $h_w(\cdot)$ has a root in $\mathbb R$.

It remains to show that $h_w(\cdot)$ has a root in $\mathbb R$. 
Let $v_0$ be such that $p(v_0) = q_w(v_0)$, whose existence and uniqueness are ensured by \eqref{pf-consistency-wgt-svm-aux-2}. By the definition of $v_0$, $q_w^{-1}\big( p(v_0)\big) = v_0$, i.e.,  $\mathbb E 1\{U \leq V<v_0\} = w\mathbb E 1\{ U>V>v_0\}$. 
Thus, 
\begin{align}\label{pf-foc-c-r-soln-aux-h-wgt-property-2}
h_w(v_0) & = \mathbb E 1\{ U \leq V<v_0\} V - w\mathbb E 1\{U>V>v_0\} V  \notag \\ 
& = \mathbb E 1\{ U \leq V<v_0\} (V -v_0) + w\mathbb E 1\{U>V>v_0\} (v_0-V) \notag \\
& < 0 
\end{align}
where the last inequality is because $V$ has positive Lebesgue density on $\mathbb R$. 
Moreover, by \eqref{pf-consistency-wgt-svm-aux-3}, 
\begin{align*}
& \lim_{v\to \infty} h_w(v) = \mathbb E 1\{U \leq V\} V - w \mathbb E 1\{U>V  \} V \\
& = \mathbb P\{U\leq V\} \mathbb E(V|U\leq V) - w \mathbb P\{U> V\} \mathbb E(V|U> V) \\
& = \mathbb P\{U\leq V\} \Big( \mathbb E(V|U\leq V) - \mathbb E(V|U> V) \Big) >0
\end{align*}
where the third equality is by $w \mathbb P\{U> V\} = w\mathbb P\{Y=-1\} = \mathbb P\{Y=1\} = \mathbb P\{U\leq V\}$. 
Together with \eqref{pf-foc-c-r-soln-aux-h-wgt-property-2} and the continuity of $h_w$, $h_w(\cdot)$ has a root in $\mathbb R$. 
This was to be shown. 
\end{proof}

\clearpage 
\bibliographystyle{econometrica}
\bibliography{bcm-svm}

\newpage

\end{document}